\documentclass[english]{article}
\usepackage[T1]{fontenc}
\usepackage[latin9]{inputenc}
\usepackage{geometry}
\usepackage{color}
\usepackage{babel}
\usepackage{float}
\usepackage{amsmath}
\usepackage{amsthm}
\usepackage{amssymb}
\usepackage{graphicx} 
\usepackage{algorithm}
\usepackage{tasks}
\settasks{style=itemize}

\usepackage{bm}
\usepackage{amsmath, amsfonts, amssymb, amsthm}
\usepackage{mathrsfs, dsfont}
\usepackage[dvipsnames]{xcolor}

\usepackage{bbold}  % or try dsfont

\newcommand{\R}{\mathbb{R}}
\newcommand{\N}{\mathbb{N}}

\newcommand{\E}{\mathbb{E}}

\def\P{\mathbb{P}}

\newcommand\F{\mathcal{F}}

\newcommand{\alphabet}[1][d]{{A}_{#1}}
\newcommand{\TA}[1][d]{T(\R^{#1})}
\newcommand{\eTA}[1][d]{T((\R^{#1}))}
\newcommand{\tTA}[2][d]{T^{#2}(\R^{#1})}
\newcommand{\bpsi}{\bm{\psi}}

\newcommand{\word}[1]{{\mathcolor{NavyBlue}{\mathbf{#1}}}}
\newcommand{\emptyword}{{\color{NavyBlue}\textup{\textbf{\o{}}}}}

\newcommand{\proj}[1]{|_{\word{#1}}}

\newcommand{\element}[1]{^{\word{#1}}}

\newcommand{\conpow}[1]{^{\otimes #1}}

\newcommand{\shuprod}{\mathrel{\sqcup \mkern -3.2mu \sqcup}}
\newcommand{\shupow}[1]{^{\shuprod #1}}

\newcommand{\A}{\mathcal{A}}

\newcommand{\norm}[1]{|| #1 ||}

\NewDocumentCommand{\sigX}{O{t} O{X}}{\mathbb{#2}_{#1}}
\NewDocumentCommand{\sig}{O{t} O{W}}{\widehat{\mathbb{#2}}_{#1}}
\NewDocumentCommand{\sigE}{O{t} O{W}}{\E[\sig[#1][#2]]}

\NewDocumentCommand{\bracketsigX}{O{t} O{X} m}{\left \langle #3, \sigX[#1][#2] \right \rangle} % not an actual error
\NewDocumentCommand{\bracketsig}{O{t} O{W} m}{\left \langle #3, \sig[#1][#2] \right \rangle}   % not an actual error
\NewDocumentCommand{\bracketsigtrunc}{O{M} O{t} O{W} m}{\left \langle #4, \sig[#2][#3]^{\leq #1} \right \rangle}   % not an actual error
\NewDocumentCommand{\bracketsigE}{O{t} O{W} m}{\left \langle #3, \sigE[#1][#2] \right \rangle} % not an actual error

\newcommand{\bell}{\bm{\ell}}
\newcommand{\bp}{\bm{p}}

\usepackage[round]{natbib}
\usepackage{enumitem}

\usepackage[unicode=true,pdfusetitle,
 bookmarks=true,bookmarksnumbered=false,bookmarksopen=false,
 breaklinks=false,pdfborder={0 0 0},pdfborderstyle={},backref=false,colorlinks=true]
 {hyperref}
\hypersetup{
 citecolor=blue, linkcolor=red}

\makeatletter

\newcommand{\noun}[1]{\textsc{#1}}
\numberwithin{equation}{section}

\theoremstyle{plain}
\newtheorem{thm}{Theorem}[section]

\newtheorem{prop}[thm]{Proposition}
\newtheorem{lem}[thm]{Lemma}

\newtheorem{example}[thm]{Example}

\theoremstyle{definition}
\newtheorem{defn}[thm]{Definition}
\theoremstyle{remark}
\newtheorem{rem}[thm]{Remark}

\makeatother

\newenvironment{sqremark}{\begin{rem}}{\hfill \tiny $\blacksquare$ \end{rem}}
\newenvironment{sqexample}{\begin{example}}{\hfill \tiny $\blacksquare$ \end{example}}

\usepackage[normalem]{ulem}

\usepackage{todonotes}
\usepackage{comment}
\usepackage{mathtools}
\usepackage{authblk}
\title{Signature approach for pricing and hedging path-dependent options with frictions }

\author[1]{Eduardo Abi Jaber\thanks{\emph{eduardo.abi-jaber@polytechnique.edu.}  EAJ is  grateful for the financial support from the Chaires FiME-FDD and  Financial Risks  at Ecole Polytechnique.}}
\author[2]{Donatien Hainaut\thanks{\emph{donatien.hainaut@uclouvain.be}}}
\author[2]{Edouard Motte\thanks{\emph{Corresponding author, edouard.motte@uclouvain.be.} EM is  grateful for the financial support from the Fonds de la Recherche
Scientifique (F.R.S. - FNRS) through a FRIA grant.}}
\affil[1]{Ecole Polytechnique, CMAP}
\affil[2]{Universit\'e Catholique de Louvain, LIDAM-ISBA}

\begin{document}
\maketitle

\renewcommand{\thefootnote}{\fnsymbol{footnote}}
\footnotetext{We would like to thank Louis-Amand G\'erard and Dimitri Sotnikov for fruitful discussions and insightful comments.}
\renewcommand{\thefootnote}{\arabic{footnote}}

\begin{abstract}
We introduce a novel signature approach for pricing and hedging path-dependent options with instantaneous and permanent  market impact under a mean-quadratic variation criterion. Leveraging the expressive power of signatures, we recast an inherently nonlinear and non-Markovian stochastic control problem into a tractable form, yielding hedging strategies in (possibly infinite) linear feedback form in the time-augmented signature of the control variables, with coefficients characterized by non-standard infinite-dimensional Riccati equations on the extended tensor algebra. Numerical experiments demonstrate the effectiveness of these signature-based strategies for pricing and hedging general path-dependent payoffs in the presence of frictions. In particular, market impact naturally smooths optimal trading strategies, making low-truncated signature approximations highly accurate and robust in frictional markets, contrary to the frictionless case.

\end{abstract}
\noun{Keywords:} {path-signatures, path-dependent options, market frictions, non-Markovian stochastic control, infinite-dimensional Riccati equations}.
\section{Introduction}

Pricing and hedging derivatives is a central problem in quantitative finance. In classical complete models, perfect replication is achieved by trading continuously in the underlying asset, assuming frictionless markets. In practice, however, large trades affect prices. Two types of market impact are typically distinguished: temporary impact, which affects the execution price but not future mid-prices, and permanent impact, which shifts the future mid-price due to supply/demand imbalances. Accounting for these frictions transforms option hedging into a stochastic control problem, where the objective boils down to tracking certain option's sensitivities (such as the delta) while balancing replication accuracy, trading costs, and market risk. The control problem usually falls outside the tractable linear-quadratic framework. \\

The majority of existing literature on pricing and hedging with frictions focuses on European payoffs. Solutions are typically characterized via Hamilton-Jacobi-Bellman (HJB) equations as in  \cite*{almgren2016option,Becherer2018hedging,Cetin2010option,Ekren2022utility,Rogers2010cost,Bouchard2017hedging,Said2019option,GueantPu2017,lions2007large,Loeper2018option}. While these frameworks provide rigorous insights for European options, they generally do not extend easily to path-dependent derivatives, such as Asian, barrier, or look-back options, whose payoffs depend on the entire trajectory of the underlying. Yet these are among the most common contracts in practice, especially in equity, FX, and commodity markets, making the gap between theory and applications particularly striking. A notable exception is the case of  Accelerated Share Repurchase contracts with an Asian-type component studied by \cite*{gueant2015accelerated,jaimungal2017optimal}, where, after approximations and dimension reduction, the problem can be made Markovian in three dimensions. Another exception is the work of  \cite*{bank2017hedging}, which considers general non-Markovian payoffs, with an instantaneous market impact component but without incorporating transient or permanent impact. In that setting, the problem falls back into the linear-quadratic class (with stochastic coefficients) and can be solved explicitly.  \\

We develop a novel signature approach for hedging path-dependent options under both temporary and permanent impact within a mean-quadratic variation framework as in  \cite{almgren2016option}. \cite{almgren2016option} studied the problem of hedging European options,
%while \cite*{bank2017hedging} obtained explicit solutions for hedging general payoffs but assuming only temporary impact.  
and relate the  optimal hedging strategies to an HJB  equation. They obtain explicit solutions for European options with a constant Gamma.   We extend their setting to path-dependent options,  namely with  signature payoffs. 
Because the problem is inherently non-Markovian, the classical HJB approach is no longer applicable, which is the first major challenge. To overcome this, we lift the hedging problem into the space of path-signatures, which restores Markovianity. In principle, this leads to an infinite-dimensional HJB equation, but such an equation is numerically intractable, especially since permanent impact makes the control problem highly nonlinear and far beyond the linear-quadratic class.\\

The key insight is that, for polynomial European payoffs, the HJB equation decouples into a system of non-standard infinite-dimensional Riccati ordinary differential equations (ODE), thanks to the algebraic structure of polynomials, see Section~\ref{S:polynomial}. Crucially, this algebraic structure is inherited by signatures, since signatures extend polynomials to path space. We leverage this to transform the nonlinear, path-dependent control problem into an infinite-dimensional yet numerically tractable Markovian framework. We relate the solution  to an infinite-dimensional Riccati ODE system on the extended tensor algebra with an optimal hedging strategy given in  linear feedback form in the time-extended  signature of the control variables, see Theorem~\ref{thm: verification_result}. In short, the time-extended signatures of the control variables serve not only as the state variables of the control problem but also effectively linearize it leading to Riccati-type structures familiar from classical linear-quadratic control.
 \\

Our contribution falls naturally within the recent line of work on rough paths and signature methods, which provide hierarchical, finite-dimensional representations for encoding path-dependence.
 In mathematical finance, path-signatures, introduced by  \cite{chen1957integration, chen1977iterated}, were recently used to solve different problems such as:
\begin{itemize}
			\item {{Pricing and hedging}:} \cite*{abi2025signature, abijaber2025hedging, bayer2025pricing, cuchiero2025joint, lyons2020non},
			\item {{Optimal stopping}:} \cite*{bayer2023optimal, bayer2025primal}, 
			\item  {{Portfolio allocation and trading}:} \cite*{cuchiero2025signature,  futter2023signature, kalsi2020optimal}.
\end{itemize}
Moreover, \cite*{bank2024stochastic} have recently used signatures to tackle stochastic control problems. They considered two signature-based classes of controls to approximate the class of admissible controls, showed that those classes are dense in the set of admissible controls, and proposed numerical methods to solve the problem. Compared to \cite*{bank2024stochastic}, we propose a distinct approach. Rather than restricting the optimization to a possibly smaller subset from the outset, we consider a general set of admissible progressively measurable controls. Using the standard martingale optimality principle, we then derive a verification result expressed in terms of an infinite-dimensional system of Riccati equations, which yields the optimal control in feedback form as a (possibly infinite) linear combination of time-extended signature elements, with the time-dependent coefficients given by the solutions of the Riccati equations. In the literature, for instance in  \cite*{lyons2020non, bank2024stochastic}, only time-independent coefficients have been considered so far when parameterizing the solutions in terms of time-extended signatures. This time-independence implies an approximation by controls that are analytic in the time variable, a property that the true optimal control generally does not satisfy. This provides one explanation for why our approach  achieves greater accuracy and numerical stability. \\

For the first time in the context of stochastic optimal control, we establish this connection between non-Markovian control problems and infinite-dimensional Riccati equations on the extended tensor algebra.  We note, however, that similar infinite-dimensional systems of Riccati equations have previously appeared in the literature for solving uncontrolled mathematical finance problems using signatures,  to compute the characteristic function associated with signature volatility models  and signature SDEs, see \cite*{ abi2025signature, abijaber2024fourier, cuchiero2023signature}. A general existence and uniqueness theory for such Riccati equations and the associated convergence of the power expansions appears to be intricate and remains an open problem, with certain partial results available for the Riccati equations in \cite*{abijaber2024fourier, cuchiero2023signature} and for the power series in \cite*{abi2024path, arous1989flots}.\\

We show that our signature approach offers a flexible, systematic, and computationally feasible route to manage path-dependent payoffs in markets with frictions. In particular, our  contributions can be summarized as follows:
\begin{itemize}
    \item {\textbf{Perfect hedging of signature payoffs in a frictionless market}: We derive an explicit solution for the fair price and the perfect hedging strategy for signature payoffs using the \cite{fawcett} formula, as stated in Theorem~\ref{thm: perfect_hedging}.}

    \item {\textbf{Existence and uniqueness for the stochastic control problem with frictions}: We use a variational approach to derive the existence and uniqueness of a solution to the stochastic control problem of hedging signature payoffs with market frictions in a mean-variation quadratic framework. More precisely, Theorem \ref{thm: mono_cond_existence_unicity} states that under a monotonicity condition \eqref{eq:monotonicity_cond} on the G\^ateaux derivative of the gain functional, there exists a unique maximizer of the problem, which is also the unique solution of the first-order condition \eqref{eq: FOC}. We also provide two examples (see Proposition \ref{prop: example_monotonicity_ok_without_permanent} and Proposition  \ref{prop: example_monotonicity_ok}) where the monotonicity condition is fulfilled.}
    \item {\textbf{Verification result for the stochastic control problem with frictions}:} By combining the It\^o formula for signatures (Theorem \ref{thm:sig_ito}) and a square completion approach, we derive a general verification theorem using the martingale optimality principle, see Theorem \ref{thm: verification_result}. Under some assumptions, including the existence of a suitable solution to an infinite-dimensional system of Riccati equations given by \eqref{eq:infinite_dim_riccati_psi_t} and the convergence of the related infinite series, this theorem explicitly characterizes the optimal solution of the stochastic control problem in a feedback form as a linear combination of time-extended signature elements with time-dependent coefficients, see \eqref{eq:hedgingoptimal_speed}. Moreover, we provide two concrete examples for which we can establish the existence of an explicit solution to the system of Riccati equations: 
    \begin{enumerate}[label=(\roman*)]
        \item Proposition \ref{prop: explicit_sol_without_permanent_impact} considers the case without permanent impact and allows us to recover the result of \cite*[Theorem 3.1]{bank2017hedging}, but applied to signature payoffs.
        \item Proposition \ref{prop: explicit_sol_eu_quad} considers European quadratic payoff, and we retrieve the result of \cite*[Eq.~(25)]{almgren2016option}.
    \end{enumerate}

    \item {\textbf{Numerical illustrations}}: In Section~\ref{sec: numerical_results}, we illustrate numerically our signature-based hedging strategies,  given by Theorem~\ref{thm: verification_result} by solving the system of infinite-dimensional Riccati equations \eqref{eq:infinite_dim_riccati_psi_t}, demonstrating that our approach can be effectively implemented even in highly path-dependent settings. First, we consider exact signature payoffs, such as European or Asian quadratic options, using sanity checks against explicit solutions to validate both the assumptions of Theorem~\ref{thm: verification_result} and our numerical implementation. Second, we consider more general path-dependent payoffs, which we approximate by signature payoffs relying on the universal approximation theorem.  In the frictional market, we show that the signature-based strategy computed via the Riccati system for the approximate signature payoffs provides a sufficiently accurate approximation of the true optimal strategy, demonstrating a striking improvement compared to signature-based strategies in the frictionless setting. This improvement can be attributed to the fact that market impact naturally smooths the optimal trading strategies, making polynomial approximations more accurate and efficient.
  We also evaluate the impact of the permanent market impact component on both the fair price and the hedging strategies of several path-dependent options. 
\end{itemize}

The paper is outlined as follows. In Section \ref{sec: bachelier_framework}, we introduce the Bachelier framework with both temporary and permanent market impact. Then, Section \ref{sec:Prel_signature} recalls key results on path-signatures. In Section \ref{sec:perfect_hedging_frictionless_market}, we first address the hedging of a signature payoff in frictionless market and we deduce the perfect hedging strategy using the \cite{fawcett} formula. Section \ref{sec: MQV_hedging_impact} is devoted to the stochastic control problem of hedging a signature payoff under market impact within a mean-quadratic variation framework, where we establish an existence and uniqueness result (Section \ref{sec: existence_uniqueness}) and provide a verification theorem (Section \ref{sec: verification}). Finally, Section \ref{sec: numerical_results} presents numerical illustrations.

\section{Bachelier model with temporary and permanent impact}\label{sec: bachelier_framework}We fix a finite horizon $T>0$ and a filtered probability space $(\Omega,\mathcal{F},(\mathcal{F}_{t})_{0\leq t\leq T},\mathbb{P})$
with filtration $(\mathcal{F}_{t})_{0\leq t\leq T}$ that satisfies the usual conditions. We consider a financial market with a tradable asset $(S_t)_{t\in[0,T]}$ with dynamic
\begin{align}
    dS_t = \mu dt + \sigma dW_t,~t\leq T,
\end{align}
 with $\mu\in\mathbb{R},$ $\sigma\in\mathbb{R}_+$ and $W$ a standard Brownian motion. In this financial market, a trader shorts an option on $S$ that she aims to hedge. We assume that by trading the asset, the trader incurs both temporary and permanent market impacts which affect her profit and loss (P\&L). More precisely, if $(X_t^\theta)_{t\in[0,T]}$ denotes the number of shares held by the trader through time defined as 
\begin{equation}\label{eq: hedging_strat_trading_speed}
    X_t^\theta = X_0+\int_0^t \theta_s ds, 
\end{equation}

with $X_0\in \mathbb{R}$ \footnote{ As explained in \cite{GueantPu2017}, $X_0$ corresponds to the number of shares in the hedging portfolio at inception. In illiquid markets, the client that buys the option may provide an initial number of shares (see Section \ref{sec: MQV_hedging_impact}). Therefore, we could consider $X_0=0$ or the case where $X_0$ is equal to the Bachelier $\Delta$ of the option sold, depending on the initial agreement with the client.}, and $(\theta_t)_{t\in[0,T]}$ the trading speed, the price of the asset, taking account of permanent impact, is given by 
\begin{equation}\label{eq: public_traded_price}
    P_t^{\theta} = S_t + \nu (X_t^\theta-X_0),\quad t\leq T,
\end{equation}
with $\nu \geq 0, $ the permanent market impact parameter. Moreover, as the trader incurs transaction costs (temporary impact), the effective traded asset price is 
\begin{equation}
    \label{eq: private_traded_price} \tilde{P}_t^{\theta}=P_t^{\theta} + \eta\theta_t,\quad t\leq T,
\end{equation}
with $\eta\geq 0,$ the temporary price impact parameter. \\

\subsection{The Almgren and Li setup}
Before addressing the problem of hedging signature payoffs with market impacts, and to get more insight about our motivation to use a signature approach in this context,  let us start from the Markovian framework of \cite{almgren2016option}. For this, we fix a European option  whose frictionless Bachelier delta at time $t$ is $\Delta(t,S_t)$ and we  consider the problem of  hedging this option with temporary and permanent market impact. This essentially boils down to tracking  the Bachelier delta $\Delta(t,P^\theta_t)$ evaluated with the impacted price $P^\theta$ given by \eqref{eq: public_traded_price}. \cite{almgren2016option} consider a mean-quadratic variation criteria on the $\mbox{P\&L}^{\theta}$ generated by the trading strategy $\theta$ of the form\footnote{The mean-quadratic variation hedging problem is introduced in detail in Section \ref{sec: MQV_hedging_impact}}
$$\E\left(\mbox{P\&L}_T^{\theta}-\frac{\lambda}{2}[\mbox{P\&L}^\theta, \mbox{P\&L}^\theta]_T\right),$$
where $\lambda>0$ is the risk aversion of the trader. The aim is to maximize the criterion over the   trading strategies $\theta$. Using a dynamic programming approach, they characterized the solution in terms of the solution to a Hamilton-Jacobi-Bellman (HJB) equation. More precisely, they showed that the value function $J(t,p,x)$ depends, at time $t<T$, on $(P^\theta_t,X^\theta_t)$ and satisfies a nonlinear HJB equation of the following form\footnote{The HJB equation differs slightly from that in \cite{almgren2016option} because, in this paper, we consider a short position to be hedged and $X^\theta$ as a state variable instead of $Y^\theta:=X^\theta-\Delta(t,P^\theta)$. Moreover, we also assume no overnight risk and $\mu\neq 0$.}\footnote{ \cite{almgren2016option} do not provide existence or well-posedness results for the HJB equation \eqref{eq: HJB_european_payoff}. For a rigorous analysis of related HJB equations, see \cite{Ekren2022utility}.}:
\begin{equation}\label{eq: HJB_european_payoff}
\begin{aligned}
    &J_t = -\frac{\lambda}{2} \sigma^2 \left[x-\Delta(t,p) \right]^2- \frac{1}{2} \sigma^2 J_{pp}+\frac{1}{4\eta} \left[ \nu (x-\Delta(t,p))-J_x - \nu J_p  \right]^2\\
    &~~~~~~~+\mu \left[(x-\Delta(t,p))-J_p \right],\\
    &J(T,p,x) =0. 
\end{aligned}
\end{equation}
 When the Bachelier Gamma $\Gamma(t,p)$ is constant, i.e. $\Gamma(t,p) = \Gamma\in \mathbb{R}$, for all $(t,p)\in [0,T] \times \mathbb{R}^+$, and $\mu=0$, they deduced that $J(t,p,x)$ has an explicit form. This explicit solution follows from the fact that, when  $\Gamma$  is constant, the problem becomes linear-quadratic: a constant gamma implies that the delta is linear in 
$p$, and therefore that the payoff is quadratic. A natural idea is thus to extend this observation to more general polynomial payoffs.

\subsection{Key observation for Markovian polynomial payoffs}\label{S:polynomial}
In fact, we can go further and obtain a semi-explicit characterization of the solution in terms of infinite dimensional Riccati ordinary differential equations when we consider polynomial payoffs of the form
\begin{equation}\label{eq: polynomial_payoff}
    H_T^\theta = \sum_{i=0}^M \alpha_i \left(P_T^\theta \right)^i,~(\alpha_i)_{i=1,..,M}\in \mathbb{R}^M,~M\in \mathbb{N}.
\end{equation}
Then, there exists $(\delta_i(t))_{i=0,...,M}\in \mathbb{R}^M$ such that
\begin{equation*}
    \Delta(t,p) = \sum_{i=0}^{M-1} \delta_i(t) p^i.
\end{equation*}
Making an ansatz on the value function of the form of a power series in $p$ and $x$:
\begin{align}\label{eq: polynomial_ansatz_value_function}
    J(t,p,x) = \sum_{i,j=0}^{+\infty} \psi_t^{i,j} p^i x^j, 
\end{align}
and plugging this ansatz in \eqref{eq: HJB_european_payoff}, the HJB equation decouples into a system of infinite-dimensional  Riccati equations given, for $i,j \in \mathbb{N}$, by 
\begin{equation}\label{eq: infinite_riccati_eq_Markovian_case}
\begin{aligned}
    &\dot{\psi}_t^{i,j} = -\frac{1}{2} \sigma^2 \psi_t^{i+2,j}-\frac{\lambda}{2} \sigma^2 \left( \boldsymbol{\beta}_t * \boldsymbol{\beta}_t\right)^{i,j}+\frac{1}{4\eta} \left(\boldsymbol{\tilde{\psi}}_t *\boldsymbol{\tilde{\psi}}_t \right)^{i,j}\\
    &~~~~~~~~~+ \mu \left[\beta_t^{i,j}-(i+1) \psi_t^{i+1,j}\right],\\
    &\psi_T^{i,j} = 0,
\end{aligned}
\end{equation}
with, $\boldsymbol{\beta}_t=\left(\beta_t^{i,j}\right)_{i,j\geq 0}$ and $\boldsymbol{\tilde{\psi}}_t=\left(\tilde{\psi}_t^{i,j}\right)_{i,j\geq 0}$ given by 
\begin{align*}
    \beta_t^{i,j}&:= \mathbb{1}_{\{i=0\} \cap \{j=1\}}-\mathbb{1}_{\{i\leq M-1\}\cap \{j=0\}} \delta_i(t), \\
    \tilde{\psi}_t^{i,j}&:= \nu\beta_t^{i,j}-\nu \psi_t^{i+1,j} (i+1)-\psi_t^{i,j+1} (j+1),
\end{align*}
and the $2D$ Cauchy product defined, for $\mathbf{a}=\left(a^{i,j}\right)_{i,j \geq 0},~\mathbf{b}=\left(b^{i,j}\right)_{i,j \geq 0}$, by
\begin{equation*}
    \left(\mathbf{a} * \mathbf{b}\right)^{i,j}:= \sum_{m=0}^i \sum_{n=0}^j a^{mn} b^{i-m,j-n}.
\end{equation*}
% then, $J$ given by \eqref{eq: polynomial_ansatz_value_function} satisfies the equation HJB \eqref{eq: HJB_european_payoff} for polynomial payoffs of the form \eqref{eq: polynomial_payoff}.
Therefore, in the Markovian case (European payoffs), by restricting ourselves to polynomial payoffs, we can semi-explicitly express the optimal solution as  a power series in the state variables with deterministic coefficients $\psi$ coming from an infinite-dimensional system of Riccati equations. Since path-signatures play a similar role to polynomials on the path space, by considering signatures payoffs, we can naturally expect to be able to derive the same kind of characterization for the control problem in the non-Markovian case, and this is precisely what is established in our main result, Theorem~\ref{thm: verification_result}.

\section{Reminder on signatures} \label{sec:Prel_signature} 
    In this section, we provide a reminder about the path-signature theory. We also refer to the first sections in \cite{lyons2020non,bayer2023optimal,cuchiero2023signature,abi2024path}.
\subsection{Tensor algebra}
    Let $d \in \N$ and denote by $\otimes$ the tensor product over $\R^d$, e.g. $(x \otimes y \otimes z)_{ijk} = x_i y_j z_k$, for $i, j, k = 1, \dots, d$, for $x, y, z \in \R^d$. For $n\geq 1$, we denote by $(\R^d) \conpow{n}$ the space of tensors of order $n$ and by $(\R^d) \conpow{0} = \R$. In this paper, we consider path-signatures that are mathematical objects that live on the extended tensor algebra space $\eTA $ over $\R^d$, that is the space of (infinite) sequences of tensors defined by
    $$ \eTA := \left\{ \bell = (\bell^n)_{n=0}^\infty : \bell^n \in (\R^d) \conpow{n} \right\}. $$
    
    For $M \geq 0$, we define the truncated tensor algebra $\tTA{M}$ as the space of sequences of tensors of order at most $M$ defined by
    
    $$ \tTA{M} := \left\{ \bell \in \eTA : \bell^n = 0, \text{ for all } n > M \right\}, $$
    
    and the tensor algebra $\TA$ as the space of all finite sequences of tensors defined by
    
    $$ \TA := \bigcup_{M \in \N } \tTA{M}. $$
    
    For $\bell = (\bell^n)_{n \in \N},~\bp = (\bp^n)_{n \in \N} \in \eTA$ and $\lambda \in \R$, we define the following operations:
    
    \begin{align*}
        \bell + \bp :&
        = (\bell^n + \bp^n)_{n \in \N}
        \\ \bell \otimes \bp :&
        = \left( \sum_{k=0}^n \bell^k \otimes \bp^{n-k} \right)_{n \in \N} \\
        \lambda \bell :&
        = (\lambda \bell^n)_{n \in \N}.
    \end{align*}
    
    These operations induce analogous operations on $\tTA{M}$ and $\TA$. \\
    
    Let $\{ e_1, \dots, e_d \} \subset \R^d$ be the canonical basis of $\mathbb{R}^d$ and $\alphabet = \{ \word{1}, \word{2}, \dots, \word{d} \}$ be the corresponding alphabet. To ease reading, for $i \in \{ 1, \dots, d \}$, we write $e_{i}$ as the blue letter $\word{i}$ and for $n \geq 1, i_1, \dots, i_n \in \{ 1, \dots, d \}$, we write $e_{i_1} \otimes \cdots \otimes e_{i_n} $ as the concatenation of letters  $\word{i_1 \cdots i_n}$, that we call a word of length $n$. We note that $(e_{i_1} \otimes \cdots \otimes e_{i_n})_{(i_1, \dots, i_n) \in \{ 1, \dots, d \}^n}$ is a basis of $(\R^d) \conpow{n}$ that can be identified with the set of words of length $n$ defined by 
    
    \begin{equation} \label{eq:sig_basis}
        V_n := \{ \word{i_1 \cdots i_n}: \word{i_k} \in \alphabet \text{ for } k = 1, 2, \dots, n \}. 
    \end{equation} 
    
    Moreover, $\emptyword$ denotes the empty word and $V_0 = \{ \emptyword \}$ serves as a basis for $(\R^d) \conpow{0} = \R$. It follows that $V := \cup_{n \geq 0} V_n$ represents the standard basis of $\eTA$. In this case, every $\bell \in \eTA$, can be decomposed as
    
    \begin{equation} \label{eq:sig_expansion}
        \bell = \sum_{n=0}^{\infty} \sum_{\word{v} \in V_n} \bell^{\word{v}} \word{v},
    \end{equation}
    where $\bell^\word{v}$ is the real coefficient of $\bell$ at coordinate $\word{v}$.\\

    We now introduce two important concepts such as the concatenation and the projection. The concatenation $\bell \word{v}$ of elements $\bell \in \eTA$ and the word $\word{v} = \word{i_1 \cdots i_n}$ means $\bell \otimes e_{i_1} \otimes \cdots \otimes e_{i_n}$. The projection $\bell \proj{u} \in \eTA$ is defined as
    \begin{equation} \label{eq:projection}
        \bell \proj{u} := \sum_{n=0}^{\infty} \sum_{\word{v} \in V_n} \bell^{\word{vu}} \word{v}
    \end{equation}
    for all $\word{u} \in V$. Moreover, if we consider $\xi \in T((\mathbb{R}^d))$, then we define the projection $\bell|_{\xi} \in T((\mathbb{R}^d))$ by
    \begin{equation}\label{eq: proj_tensor}
        \bell|_{\xi} := \sum_{n=0}^{\infty} \sum_{\word{v} \in V_n}\left[\sum_{m=0}^{\infty}\sum_{\word{u}\in V_m}\bell^{\word{vu}} \xi^{\word{u}} \right]\word{v}. 
    \end{equation}
    
    The projection plays an important role in the space of iterated integrals as it is closely linked to partial differentiation, in contrast with the concatenation that relates to integration. We now define the bracket between $\bell \in T(\mathbb{R}^d)$, and $\bp \in \eTA$ by
    
    \begin{align} \label{eq:bracket}
        \langle \bell, \bp \rangle 
        := \sum_{n=0}^{\infty} \sum_{\word{v} \in V_n} \bell^{\word{v}} \bp^{\word{v}}. 
    \end{align}
    We observe that the bracket is well defined since $\bell\in T(\mathbb{R}^d)$ and has therefore finitely many non-zero terms. For $\bell \in \eTA$, the series in \eqref{eq:bracket} involves infinitely many terms and requires special care to be well defined. We discuss it in Subsection~\ref{S:infinite}. Finally, we consider another operation on the space of words which is the shuffle product. The shuffle product plays a crucial role for an integration by parts formula on the space of iterated integrals, see Proposition~\ref{prop:shufflepropertyextended} below.
    
    \begin{defn}[Shuffle product] \label{def:shuffleprod}
        The shuffle product $\shuprod: V \times V \to \TA$ is defined inductively for all words $\word{v}$ and $\word{w}$ and all letters $\word{i}$ and $\word{j}$ in $\alphabet$ by
        
        \begin{align*}
            (\word{v} \word{i}) \shuprod (\word{w} \word{j}) &
            = (\word{v} \shuprod (\word{w} \word{j})) \word{i} + ((\word{v} \word{i}) \shuprod \word{w}) \word{j},
            \\ \word{w} \shuprod \emptyword &
            = \emptyword \shuprod \word{w} = \word{w}.
        \end{align*}
        
        With some abuse of notation, the shuffle product on $\eTA$ induced by the shuffle product on $V$ is also be denoted by $\shuprod$.
    \end{defn}

\subsection{Signatures}

    We define the (path) signature of a semimartingale process as the sequence of iterated stochastic integrals in the sense of Stratonovich. Notice that, in the paper, we denote the It\^o integral by $\int_0^\cdot Y_t dX_t$ and the Stratonovich integral by $\int_0^\cdot Y_t \circ dX_t.$ 

    \begin{defn}\label{def:sig}
        Fix $T > 0$. Let $(X_t)_{t \geq 0}$ be a continuous semimartingale in $\R^d$ on some filtered probability space $(\Omega, \F, (\F_t)_{t \geq 0}, \P)$. The signature of $X$ is defined by  
        \begin{align*}
            \mathbb{X}: \Omega \times [0, T] &
            \to \eTA
            \\ (\omega, t) &
            \mapsto \sigX (\omega) := (1, \sigX^1(\omega), \dots, \sigX^n(\omega), \dots),
        \end{align*}
        where
        $$ \sigX^n := \int_{0 < u_1 < \cdots < u_n < t} \circ d X_{u_1} \otimes \cdots \otimes \circ d X_{u_n} $$
        takes value in $(\R^d)^{\otimes n}$, $n \geq 0$. Similarly, the truncated signature of order $M \in \N$ is defined by
        \begin{align} \label{def:sig-trunc}
            \mathbb{X}^{\leq M}: [0, T] &
            \to \tTA{M}
            \\ (\omega, t) &
            \mapsto \sigX^{\leq M}(\omega) := (1, \sigX^1(\omega), \dots, \sigX^M(\omega), 0, \dots, 0, \dots).
        \end{align}
    \end{defn} 
    
    The signature plays a similar role to polynomials on path-space. Indeed, in dimension $d=1$, the signature of $X$ is the sequence of monomials $\left( \frac{1}{n!} (X_t - X_0)^n \right)_{n \in \N}$. In particular, any finite combination of elements of the signature $\bracketsigX{\bell}$, defined in \eqref{eq:bracket} for $\bell \in \tTA{M}$, is a polynomial of degree $M$ in $X_t$. For example, if $d=2$ and $X_t = \widehat{W}_t := (t, W_t)$ where $W$ is a 1-dimensional Brownian motion, the first few signature orders are given by
    \begin{equation}
        \sig^0 = 1,
        \quad
        \sig^1 =
        \begin{pmatrix}
            t \\
             W_t
        \end{pmatrix},
        \quad
        \sig^2 =
        \begin{pmatrix}
            \frac{t^2}{2!} & \int_0^t s d W_s \\
            \int_0^t W_s ds & \frac{W_t^2}{2!}
        \end{pmatrix}
    \end{equation}
    and
    \begin{equation}
        \sig^3 =
        \begin{pmatrix}
            \frac{t^3}{3!} & \int_0^t \frac{s^2}{2!} d W_s & \\
            \int_0^t \int_0^s u d W_u ds & \int_0^t \int_0^s u d W_u \circ d W_s & \\
            & \int_0^t \int_0^s W_u du ds & \int_0^t \int_0^s W_u du \circ d W_s \\
            & \int_0^t \frac{W_s^2}{2!} ds & \frac{W_t^3}{3!}
        \end{pmatrix}.
    \end{equation}

\subsection{Infinite linear combinations of signature elements}\label{S:infinite}
    In this section, we recall some results on infinite linear combinations $\langle{\bell}, \mathbb{X}_t\rangle$ for certain admissible $\bell \in \eTA$ for which the infinite series makes sense. The crucial ingredients for this paper are the shuffle product (Proposition~\ref{prop:shufflepropertyextended}) and an It\^o's formula (Theorem~\ref{thm:sig_ito}). Those results are based on \cite{abi2024path}, and we refer to this paper for more details.\\

    We first consider the space $\A(\mathbb{X})$ of admissible elements $\bell$ below, using the associated semi-norm:
    
    $$ \norm{\bell}_t^{\mathcal{A}(\mathbb{X})} := \sum_{n=0}^\infty \left| \sum_{\word{v} \in V_n} \bell^\word{v} \mathbb{X}^\word{v}_t \right|, \quad t \geq 0, $$
    
    recall the definition of $V_n$ in \eqref{eq:sig_basis} and the decomposition \eqref{eq:sig_expansion}. 
    Whenever, $\norm{\bell}_t^{\mathcal{A}(\mathbb{X})} < \infty$ a.s., the infinite linear combination 
    
    $$ \langle{\bell}, {\mathbb{X}}_t\rangle = \sum_{n=0}^\infty \sum_{\word{v} \in V_n} \bell^{\word{v}} \mathbb{X}^{\word{v}} $$
    
    is well defined. This leads to the following definition for the admissible set $\mathcal{A}(\mathbb{X})$:
    
    $$ \mathcal{A}(\mathbb{X}) := \left\{ \bell \in \eTA[d] : \norm{\bell}_t^{\mathcal{A}(\mathbb{X})}< \infty \text{ for all } t \in [0, T] \text{ a.s.} \right\}. \label{eq:defA} $$ 
    
    Note that $\TA[d] \subset \mathcal{A}(\mathbb{X})$ and that $\langle{\bell}, \mathbb{X}_t \rangle$ is an extension of \eqref{eq:bracket}, as the two bracket operations $\langle \cdot, \cdot \rangle$ coincide whenever $\bell \in \TA[d]$. The admissible set $\mathcal{A}(\mathbb{X})$ enables us to linearize polynomials on infinite linear combination of the signature. This is what is commonly referred to in the literature as the linearization power of the signature.

    \begin{prop}[Shuffle property] \label{prop:shufflepropertyextended}
        If $\bell_1, \bell_2 \in \mathcal{A}(\mathbb{X})$, then $\bell_1 \shuprod \bell_2 \in \mathcal{A}(\mathbb{X})$ and
        $$ \langle {\bell_1}, {\mathbb{X}}_t\rangle \langle{\bell_2},{\mathbb{X}}_t\rangle = \langle\bell_1 \shuprod \bell_2, {\mathbb{X}}_t\rangle. $$
    \end{prop}

    By the definition of the admissible set $\mathcal{A}(\mathbb{X})$, we know that, for elements $\bell \in \mathcal{A}(\mathbb{X})$, the process $(\langle \bell, \mathbb{X}_t \rangle)_{t \leq T}$ is well defined. This naturally raises the question of when it is a semimartingale and how to obtain its It\^o decomposition. For elements $\bell$ in the set $\mathcal{I}(\mathbb{X})$ defined by
    \begin{equation} \label{eq:I}
        \mathcal{I}(\mathbb{X}) := \left\{  \begin{array}{l}
        \bell \in \mathcal{A}(\mathbb{X}):\\
        \text{for every } \word{i} \in A_d \text{ and } \word{j} \in \mathcal{Q}_\word{i}(X),~\bell\proj{i}, \bell\proj{ji}\in \mathcal{A}(\mathbb{X}) \text{ and } \\ \int_0^T \left( \norm{\bell\proj{i}}_t^{\mathcal{A}(\mathbb{X})} d|X^\word{i}|_t + \norm{\bell\proj{ji}}_t^{\mathcal{A}(\mathbb{X})} d[X^\word{i},X^\word{j}]_t +\left(\norm{\bell\proj{i}}_t^{\mathcal{A}(\mathbb{X})}\right)^2 d[X^\word{i},X^\word{j}]_t \right)< \infty \text{ a.s.} \end{array} \right \},
    \end{equation}
    where $\mathcal{Q}_{\word{i}}(X)$ defines the set of coordinates of $X$ that have a
    non-zero quadratic covariation with $X^\word{i}$
    \begin{align*}
        \mathcal{Q}_{\word{i}}(X):=\left\{\word{j}\in A_d: [X^\word{j},X^\word{i}]_t \neq 0, \text{on a set of non-zero } dt \otimes d\P \text{ measure}\right\}, 
    \end{align*}
    then, we can prove that $(\langle \bell, \mathbb{X}_t \rangle)_{t \leq T}$ is a semimartingale and we can compute its It\^o decomposition. Notice that, for $M>0$, we have $T^M(\mathbb{R}^d)\in \mathcal{I}(\mathbb{X}) $, i.e.~showing that finite linear combinations of the signature are always semimartingales. More generally, we recall the result for time-dependent linear combinations $(\langle \bell_t, \mathbb{X}_t \rangle)_{t \leq T}$ with $\bell:[0,T] \to \mathcal{A}(\mathbb{X})$ in the set
    \begin{equation} \label{eq:Iprime}
        \mathcal{I}^{'}(\mathbb{X}) := \left\{~\begin{array}{l}
         \bell: [0, T] \to \mathcal{A}(\mathbb{X}) :\\
            \text{for every } \word{v}\in V,\: \bell^\word{v} \in C^1([0, T],\mathbb{R}),\\
            \text{for every } \word{i} \in A_d \text{ and } \word{j} \in \mathcal{Q}_\word{i}(X) \text{ and for all } t\in[0,T],~\bell_t\proj{i}, \bell_t\proj{ji}, \dot{\bell_t}\in \mathcal{A}(\mathbb{X}) \text{ and }
            \\ \int_0^T \left( \norm{\bell_t\proj{i}}_t^{\mathcal{A}(\mathbb{X})} d|X^\word{i}|_t + \norm{\bell_t\proj{ji}}_t^{\mathcal{A}(\mathbb{X})} d[X^\word{i},X^\word{j}]_t +\left(\norm{\bell_t\proj{i}}_t^{\mathcal{A}(\mathbb{X})}\right)^2 d[X^\word{i},X^\word{j}]_t \right) < \infty
        \end{array}
        \right \},
    \end{equation}
    where $\dot{\bell_t} :=\sum_{n=0}^\infty \sum_{\word{v} \in V_n} \frac{d}{dt} \bell_t^\word{v} \word{v}$ for all $t \in [0, T]$.
    
    \begin{thm}[It\^o's formula] \label{thm:sig_ito}
        Let $\bell \in \mathcal{I}(\mathbb{X}) $, then $\langle \bell, {\mathbb{X}}_t\rangle$ is an It\^o process such that
        \begin{align} \label{eq:ItoW}
            \langle{\bell}, {\mathbb{X}}_t\rangle = \bell^{\emptyword}+\sum_{i\in A_d} \int_0^t \langle{\bell \proj{i}, \mathbb{X}_s\rangle dX_s^i + \tfrac{1}{2} \sum_{i\in A_d} \sum_{j\in A_d} \int_0^t \langle\bell \proj{ji}}, {\mathbb{X}}_s \rangle d[X^j,X^i]_s,~t\leq T.
        \end{align} 
        Moreover, if $\bell \in \mathcal{I}^{'}(\mathbb{X}) $, then 
        \begin{align} \label{eq:ItohatW}
            \langle{\bell_t}, {\mathbb{X}}_t\rangle = \bell_t^{\emptyword}+\sum_{i\in A_d} \int_0^t \langle{\bell_s \proj{i}, \mathbb{X}_s\rangle dX_s^i + \tfrac{1}{2} \sum_{i\in A_d} \sum_{j\in A_d} \int_0^t \langle\bell_s \proj{ji}}, {\mathbb{X}}_s \rangle d[X^j,X^i]_s + \int_0^t \langle\dot{\bell}_s, \mathbb{X}_s\rangle ds,~t \leq T.
        \end{align} 
    \end{thm}
    \begin{proof}
        We refer to \cite[Theorem 3.3. and Corollary 3.4.]{abi2024path} for the proof.
    \end{proof}

\section{Pricing and perfect hedging in frictionless market}\label{sec:perfect_hedging_frictionless_market}

In this section, we first show how to price and perfectly replicate a signature payoff in a frictionless market. This perfect hedging strategy will enter in the hedging problem with market impact in Section \ref{sec: MQV_hedging_impact}. In the absence of market impact, i.e.~$\eta=\nu=0$, the asset traded price is $S_t$. 
In this context, we consider that the trader wants to hedge a given signature payoff
\begin{equation}\label{eq:signature_payoff_without_impact}
    H_T=\left\langle \xi,\hat{\mathbb{S}}_T \right\rangle,
\end{equation} 
for a certain admissible $\xi \in T(\mathbb{R}^2)$, hence $\xi$ has a finite number of non-zero terms. We start by giving some examples of path-dependent payoffs that can be exactly represented as signature payoffs. 

\begin{sqexample} Signature payoffs include:
    \begin{itemize}
        \item European polynomial payoffs of the form $H_T=\sum_{k=0}^N \alpha_k(S_T-K)^k$, for given $N\in \mathbb{N}$, $K\in \mathbb{R}$ and $(\alpha_k)_{k\leq N}\in\mathbb{R}^N$, since $H_T=\left\langle\sum_{k=0}^N\alpha_k\left( \word{2}+\emptyword (S_0-K)\right)\shupow{k},\hat{\mathbb{S}}_T\right\rangle$;
        \item Asian polynomial payoffs of the form  $H_T=\sum_{k=0}^N\alpha_k\left(\frac{1}{T}\int_0^T S_t dt -K\right)^k$, for given $N\in \mathbb{N}$, $K\in \mathbb{R}$ and $(\alpha_k)_{k\leq N}\in\mathbb{R}^N$, since $H_T=\left\langle\sum_{k=0}^N\alpha_k\left(\frac{1}{T}\word{21}-\emptyword K\right)\shupow{k},\hat{\mathbb{S}}_T\right\rangle$. 
    \end{itemize}
\end{sqexample}

Let $(V_t^X)_{t\in[0,T]}$ denote the hedging portfolio of the trader such that $$dV_t^X= X_t dS_t,\quad V_0^X = V_0 \in \mathbb{R}^+.$$ where $X$ is an admissible hedging in
\begin{equation*}
    \mathcal{L}^2:=\left\{X:[0,T]\times\Omega\rightarrow\mathbb{R}~\text{prog. measurable process such that } \E\left(\int_0^T X_t^2 dt\right)< \infty \right\}.
\end{equation*}
The fair pricing of the option corresponds to $V_0.$ For $\bell \in T((\mathbb{R}^d))$, let us define the tensor product exponential as
\begin{equation*}
    \exp_{\otimes}(\bell):=\sum_{k \geq 0} \frac{\bell^{\otimes k}}{k!}.
\end{equation*}
In the frictionless market, a perfect hedging strategy can be deduced using the quantity 
\begin{equation}\label{eq: def_xi_t}
    \xi_t:= \xi|_{\mathcal E_{T-t}}, \quad t\leq T,
\end{equation}
recall the projection defined by \eqref{eq: proj_tensor}, where
\begin{align}\label{eq: expectation_time_augmented_sig}
 \mathcal E_t = \exp_{\otimes}\left( \left(\word{1} +\frac{\sigma^2}2 \word{22}  \right)  t\right)  
\end{align}
corresponds to the expectation of the signature of the time-augmented stock price $S$ under the risk-neutral probability measure thanks to the  \cite{fawcett} formula. 

\begin{thm}\label{thm: perfect_hedging}
      Let $\xi \in T(\mathbb{R}^2)$. Then, $\left(\xi_t \right)_{t\leq T}$ defined by \eqref{eq: def_xi_t}is well defined for all $t\leq T$ and  belongs to $\mathcal I'(\widehat{\mathbb S})$. Furthermore, the  option \eqref{eq:signature_payoff_without_impact}  can be perfectly hedged using the portfolio with initial wealth  $V_0=\xi_0^\emptyword$ and strategy $X_t = \left\langle \xi_{t} \proj{2},\hat{\mathbb{S}}_t \right\rangle$, that is 
    \begin{subequations}
    \begin{align}
    &V^X_T = H_T =  \left\langle \xi,\hat{\mathbb{S}}_T\right\rangle, \label{eq: value_at_T_port_perfect}\\
    \intertext{In particular,}
    &V^X_t = \langle\xi_t,\hat{\mathbb{S}}_t\rangle,\quad t\leq T. \label{eq: value_at_t_port_perfect}
    \end{align}
    \end{subequations}
  
\end{thm}

\begin{proof}
    The proof follows from two following lemmas. By combining Lemma \ref{lem: pde_perfect_hedging} and Lemma \ref{lem: ode_xi_t}, we deduce that the signature payoff $\langle \xi, \hat{\mathbb{S}}_T\rangle$ can be perfectly replicated using the portfolio with initial wealth $V_0=\xi_0^\emptyword$ and strategy $X_t = \left\langle \xi_{t} \proj{2},\hat{\mathbb{S}}_t \right\rangle$.
\end{proof}

\begin{lem}\label{lem: pde_perfect_hedging}
  Let $\xi \in T(\mathbb{R}^2)$.  Assume that there exists $(\xi_t)_{t\in [0,T]}   \in \mathcal I'(\widehat{\mathbb S})$ satisfying the following equation
\begin{align}\label{eq: pde_perfect_hedging}
    \dot{\xi}_t = -\xi_{t}\proj{1} - \frac{1}{2} \sigma^2 \xi_{t}\proj{22}, \quad  \xi_T = \xi.
\end{align}
Then, the  option \eqref{eq:signature_payoff_without_impact}  can be perfectly hedged using the portfolio with initial wealth  $V_0=\xi_0^\emptyword$ and strategy $X_t = \left\langle \xi_{t} \proj{2},\hat{\mathbb{S}}_t \right\rangle$, that is \eqref{eq: value_at_T_port_perfect} and \eqref{eq: value_at_t_port_perfect} are satisfied.
\end{lem}

\begin{proof}
First, by applying the It\^o's formula in Theorem \ref{thm:sig_ito} to $\left(\langle \xi_t, \hat{\mathbb{S}}_t\rangle \right)_{t\leq T}$ and using \eqref{eq: pde_perfect_hedging}, we have that 
\begin{align*}
    \langle\xi_T,\hat{\mathbb{S}}_T\rangle &= \xi_0^{\emptyword}+\int_0^T \langle\dot{\xi_t} +\xi_{t}\proj{1}+\mu~\xi_t\proj{2} + \frac{1}{2} \sigma^2 \xi_{t}\proj{22}, \hat{\mathbb{S}}_t\rangle dt+ \int_0^T \sigma~\langle\xi_t\proj{2}, \hat{\mathbb{S}}_t \rangle dW_t, \\
    &=\xi_0^{\emptyword}+\int_0^T \langle\mu~\xi_t\proj{2}, \hat{\mathbb{S}}_t\rangle dt+ \int_0^T \sigma~\langle\xi_t\proj{2}, \hat{\mathbb{S}}_t \rangle dW_t.
\end{align*}
But as 
\begin{align*}
    V_T^X= V_0 + \int_0^T \mu X_t dt + \sigma \int_0^T X_t dW_t,  
\end{align*}
we observe that by taking
\begin{align*}
    X_t =\langle \xi_{t} \proj{2},\hat{\mathbb{S}}_t \rangle \text{ and } V_0=\xi_0^\emptyword,
\end{align*}
then, we directly obtain \eqref{eq: value_at_t_port_perfect} and \eqref{eq: value_at_T_port_perfect}.
\end{proof}

\begin{lem}\label{lem: ode_xi_t}
Let $\xi \in T(\mathbb{R}^2)$. Then, $(\xi_t)_{t\leq T}$ given by \eqref{eq: def_xi_t} is well defined, belongs to $\mathcal I'(\widehat{\mathbb S})$ and satisfies equation \eqref{eq: pde_perfect_hedging}.
\end{lem}

\begin{proof}
    First, since $\xi \in T(\mathbb{R}^2)$, we know that there exists $M>0$ such that $\xi\in T^M(\mathbb{R}^2)$. Thus, we deduce that $\xi_t$ has finitely many non-zero terms and is well defined, such that $\xi_t\in \mathcal{I}'(\hat{\mathbb{S}})$. Then, as $\xi_t=\xi|_{\mathcal{E}_{T-t}}$, we have that
\begin{equation*}
    \xi_t = \sum_{n=0}^{M} \sum_{\word{v}\in V_n} \xi^{\word{v}}_t \mathbf{v},
\end{equation*}
with, for each $\word{v}\in V_n,~n\leq M$,
\begin{equation*}
    \xi_t^\word{v} = \sum_{m=0}^{M-n} \sum_{\word{w}\in V_m} \xi^{\word{vw}} \mathcal{E}_{T-t}^\word{w}. 
\end{equation*}
Since 

\begin{equation*}
    \dot{\mathcal{E}}_{T-t} =- \left(\word{1}+\frac{1}{2}\word{22}\right) \otimes \mathcal{E}_{T-t},
\end{equation*}
we observe that, for each $\word{v}\in V_n,~n\leq M$,
\begin{align*}
    \dot{\xi_t}^\word{v}&=-\sum_{m=0}^{M-n}\sum_{\word{w}\in V_m} \xi^{\word{vw}} \left(\left(\word{1}+\frac{\sigma^2}{2} \word{22}\right)\otimes\mathcal{E}_{T-t}\right)^\word{w}\\
    &=-\sum_{m=0}^{M-n}\sum_{\word{w}\in V_m} \xi^{\word{vw}} \sum_{l=0}^{m}\left(\sum_{\word{w}'\in V_{l}, ~\word{1}\word{w}'=\word{w}} \mathcal{E}_{T-t}^{\word{w}'}+\frac{\sigma^2}{2} \sum_{\word{w}'\in V_{l}, ~\mathbf{22}\word{w}'=\word{w}} \mathcal{E}_{T-t}^{\word{w}'}\right)\\
    &=-\sum_{l=0}^{M-n}\sum_{\word{w}\in V_l} \left(\xi^{\word{v1w}} + \frac{\sigma^2}{2} \xi^{\word{v22w}} \right) \mathcal{E}_{T-t}^\word{w}\\
    &=-\left(\xi_t^{\word{v1}} +\frac{\sigma^2}{2} \xi_t^\word{v22} \right).  
\end{align*}
Therefore, we infer that
\begin{align*}
    \dot{\xi}_t &= \sum_{n=0}^M \sum_{\word{v}\in V_n} \dot{\xi}_t^\word{v} \word{v}=-\sum_{n=0}^M \sum_{\word{v}\in V_n} \left({\xi}_t^\word{v1}+\frac{\sigma^2}{2} \xi_t^\word{v22}\right) \word{v}=-\xi_t\proj{1}-\frac{\sigma^2}{2} \xi_t\proj{22},
\end{align*}
which gives \eqref{eq: pde_perfect_hedging}.
\end{proof}

\section{Mean-quadratic variation hedging with market impact}\label{sec: MQV_hedging_impact}
Let us now consider the hedging problem of signature payoffs by assuming both temporary and permanent market impact as in \eqref{eq: public_traded_price}  and \eqref{eq: private_traded_price}. In this case, the trader aims  to optimally hedge a signature payoff given by 
\begin{equation}\label{eq:signature_payoff_permanent_impact}
    H_T^\theta=\left\langle \xi,\hat{\mathbb{P}}^\theta_T \right\rangle, 
\end{equation} 
with $\xi \in T^M(\mathbb{R}^2)$, for some $M>0$. Since we assume permanent price impact, the signature payoff \eqref{eq:signature_payoff_permanent_impact} depends now on the time-augmented signature of $P^\theta$ given by \eqref{eq: public_traded_price} instead of the time-augmented signature of $S$ (compare with \eqref{eq:signature_payoff_without_impact}). In addition, the hedging strategy will be of finite variation in the form \eqref{eq: hedging_strat_trading_speed}, thanks to the regularizing effect of the temporary impact component with coefficient $\eta >0$ present in \eqref{eq: private_traded_price}. In this frictional market, we follow \cite{almgren2016option} and we assume that the trader's P\&L is marked to market using the Bachelier option price without friction such that the P\&L at time $t\leq T$ resulting from the hedging strategy denoted by $R_t^\theta$ is given by
\begin{equation}\label{eq:PL_at_t}
    R_t^\theta=(V_0-X_0S_0)+X_t^\theta P_t^\theta-\int_0^t \tilde{P}^\theta_s \theta_s ds -\langle \xi_t,\hat{\mathbb{P}}_t^\theta\rangle,~t\leq T,
\end{equation}
with $V_0$ the initial value of the hedging portfolio composed of cash and holdings $X_0$ in the traded asset, and $(\xi_t)_{t\in[0,T]}$ defined by \eqref{eq: def_xi_t}. In particular,
\begin{equation}\label{eq:PL_at_T}
    R_T^\theta=(V_0-X_0 S_0)+X_T^\theta P_T^\theta-\int_0^T \tilde{P}^\theta_t \theta_t dt -H_T^\theta.
\end{equation}

Moreover, we assume that the trader optimally hedges her short position using a mean-quadratic variation criteria in the form
\begin{equation}\label{eq:optimal_control_1}
    \sup_{\theta\in \mathcal{A}} \E\left(R_T^\theta-\frac{\lambda}{2}[R^\theta,R^\theta]_T\right), 
\end{equation}
with $\lambda>0$, the risk-aversion parameter, and $\mathcal{A}$ the set of admissible trading speed defined by

\begin{equation} 
    \mathcal{A}:=\left\{\theta:[0,T]\times\Omega\rightarrow\mathbb{R}~\text{prog. meas. process such that }\E\left(\int_0^T \theta_t^{p} dt\right)< \infty, \text{ for } p:=2\vee 2\tilde{M}\mathbb{1}_{\{\nu>0\}} \right\},
\end{equation} 
$\tilde{M}$ being the truncation order of $\xi_t\proj{2}$.\\

Using the optimal strategy from the frictionless market derived in Section \ref{sec:perfect_hedging_frictionless_market}, we can rewrite the dynamic of $R^\theta$ given by \eqref{eq:PL_at_t}, which facilitates the computation of its quadratic variation and simplifies the expression of the objective criterion.

\begin{lem}
    Let $(\xi_t)_{t\in[0,T]}$ be defined by \eqref{eq: def_xi_t} and $\theta\in \mathcal{A}$, then 
    \begin{equation}\label{eq: rewritte_PL}
        R_T^\theta = \left(V_0-\xi_0^{\emptyword} \right) + \int_0^T (\mu+ \nu\theta_t)\left(X_t^\theta-\left\langle \xi_t\proj{2},\hat{\mathbb{P}}^\theta_t \right\rangle\right)dt-\int_0^T \eta \theta_t^2 dt+\sigma\int_0^T \left( X_t^\theta - \left\langle \xi_t\proj{2},\hat{\mathbb{P}}^\theta_t \right\rangle\right) dW_t, 
    \end{equation}
    and its quadratic variation is given by 
    \begin{equation*}
        \E \left([R^\theta,R^\theta]_T\right) = \sigma^2 \E \left(\int_0^T \left( X_t^\theta - \left\langle \xi_t\proj{2},\hat{\mathbb{P}}^\theta_t \right\rangle\right)^2dt\right). 
    \end{equation*}
    The stochastic control problem can be rewritten as 
    \begin{equation} \label{eq:control_problem_rewritten}
    \begin{aligned}
    \sup_{\theta\in \mathcal{A}} \E\left(R_T^\theta-\frac{\lambda}{2}[R^\theta,R^\theta]_T\right)=& \left(V_0 -\xi_0^{\emptyword} \right)+\sup_{\theta\in\mathcal{A}} \bigg[ \E \left( \int_0^T (\mu+ \nu\theta_t)\left(X_t^\theta -\left\langle \xi_t\proj{2},\hat{\mathbb{P}}^\theta_t \right\rangle\right)dt-\int_0^T \eta \theta_t^2 dt\right)\\
    &~~~~~~~- \frac{\lambda}{2}\sigma^2 \E \bigg(\int_0^T \left( X_t^\theta - \left\langle \xi_t\proj{2},\hat{\mathbb{P}}^\theta_t \right\rangle\right)^2 dt\bigg) \bigg].
    \end{aligned}
    \end{equation}
\end{lem}
\begin{proof}
    An application of It\^{o} formula in Theorem \ref{thm:sig_ito} to $\langle \xi_t, \hat{\mathbb{P}}_t^\theta\rangle$ and \eqref{eq: pde_perfect_hedging} yields
    \begin{align*}
    H_T^\theta =\left\langle \xi_T,\hat{\mathbb{P}}^\theta_T \right\rangle =  \xi_0^{\emptyword}+ \int_0^T (\mu+\nu \theta_t)\left\langle \xi_t\proj{2},\hat{\mathbb{P}}^\theta_t \right\rangle dt + \sigma \int_0^T \left\langle \xi_t\proj{2},\hat{\mathbb{P}}^\theta_t \right\rangle dW_t,
    \end{align*}
    with $\xi_t$ defined by \eqref{eq: def_xi_t} and thus, we can rewrite $R_T^\theta$ in the form 
    \begin{align*}
    R_T^\theta =&V_0+\int_0^T \left((\mu+\nu \theta_t) X_t^\theta-\eta \theta_t^2 \right) dt + \sigma \int_0^T X_t^\theta dW_t - \left\langle \xi_T,\hat{\mathbb{P}}^\theta_T \right\rangle\\
    =&\left(V_0 -\xi_0^{\emptyword} \right) + \int_0^T (\mu+ \nu\theta_t)\left(X_t^\theta-\left\langle \xi_t\proj{2},\hat{\mathbb{P}}^\theta_t \right\rangle\right)dt-\int_0^T \eta \theta_t^2 dt+\sigma\int_0^T \left( X_t^\theta - \left\langle \xi_t\proj{2},\hat{\mathbb{P}}^\theta_t \right\rangle\right) dW_t. 
    \end{align*}
    The quadratic-variation of $R^\theta$ follows from this representation. 
\end{proof}

\begin{sqremark}
    From the reformulation of the stochastic control problem \eqref{eq:control_problem_rewritten}, we see that the problem is not Linear-Quadratic, except for the case of $\xi_t\proj{22} = \Gamma \emptyword$, $\Gamma\in \mathbb{R}$, which corresponds to the case of payoffs with constant Gamma as studied by \cite{almgren2016option}.
\end{sqremark}

Before tackling the control problem \eqref{eq:optimal_control_1}, 
it is worth first asking at what price the trader should sell the option in this market with frictions. Indeed, there is no reason for the trader to sell the option at the frictionless Bachelier price. We follow \cite{GueantPu2017} and consider an indifference pricing approach. More precisely, we assume that at time $t=0^-$ (just before the deal), the trader only holds a cash position $C\in \mathbb{R}$ and has no incentive to take a position on the risky asset. Then, the deal between the client and the trader is done as follows at $t=0$:
\begin{itemize}
    \item the trader writes the option with payoff $\langle\xi,\hat{\mathbb{P}}_T^\theta \rangle$ and the client pays a price $\pi$;
    \item the client gives $X_0$ shares to the trader and receives $X_0 S_0$ in cash from the trader.
\end{itemize}

As the trader considers a mean-quadratic variation criteria, the value she gives is the solution to the control problem \eqref{eq:optimal_control_1} with $V_0=C+\pi$. Moreover, as at inception (before the deal), the trader only holds cash $C$, her utility is simply $C$. As a consequence, the trader is indifferent to make the deal if the price $\pi$ is such that 
\begin{equation*}
    \sup_{\theta\in \mathcal{A}} \E\left(R_T^\theta-\frac{\lambda}{2}[R^\theta,R^\theta]_T\right) = C,
\end{equation*}
where $V_0 = C+\pi$. Finally, using \eqref{eq:control_problem_rewritten}, we deduce that the indifference price $\pi$ is given by
\begin{equation}\label{eq: indifference_price_sig_payoff_market_impact}
\begin{aligned}
    \pi &= \xi_0^{\emptyword} - \sup_{\theta\in\mathcal{A}} \bigg[\E \left( \int_0^T (\mu+ \nu\theta_t)\left(X_t^\theta -\left\langle \xi_t\proj{2},\hat{\mathbb{P}}^\theta_t \right\rangle\right)dt-\int_0^T \eta \theta_t^2 dt\right)\\
    &~~~~~~~~~~~- \frac{\lambda}{2}\sigma^2 \E \bigg(\int_0^T \left( X_t^\theta - \left\langle \xi_t\proj{2},\hat{\mathbb{P}}^\theta_t \right\rangle\right)^2 dt\bigg) \bigg].
\end{aligned}
\end{equation}
The indifference price is composed of two terms: the frictionless Bachelier price $\xi_0^{\emptyword}$, and a spread that depends on market impact parameters, so that the indifference price $\pi$ includes the additional risk incurred by the trader due to market frictions. In particular, when $\mu=0$ and $\nu=0$, as in \cite{GueantPu2017}, it follows from \eqref{eq: indifference_price_sig_payoff_market_impact} that $\pi\geq \xi_0^{\emptyword}$ for all $\eta\geq 0$, with $\pi=\xi_0^{\emptyword}$ for  $\eta=0$. However, when $\mu>0$ or $\nu>0$, there is no reason why this inequality should hold in all cases, as the trader could exploit a signal or an arbitrage to maximize her risk-adjusted P\&L, and therefore require an indifference price that could be lower than $\xi_0^{\emptyword}$.

\subsection{Existence and uniqueness result}\label{sec: existence_uniqueness}
In this section, we study existence and uniqueness result of the stochastic control problem \eqref{eq:control_problem_rewritten}. We consider the functional $J:\mathcal{A}\to \mathbb{R}$ by 
\begin{equation}
    J(\theta):= \E \left( \int_0^T \left[(\mu+ \nu\theta_t)\left(X_t^\theta -\left\langle \xi_t\proj{2},\hat{\mathbb{P}}^\theta_t \right\rangle\right)- \eta \theta_t^2 -\frac{\lambda}{2}\sigma^2 \left( X_t^\theta - \left\langle \xi_t\proj{2},\hat{\mathbb{P}}^\theta_t \right\rangle\right)^2 \right]dt\right), \quad \theta \in \mathcal{A}.
\end{equation}
Following \eqref{eq:control_problem_rewritten}, the stochastic control problem can be rewritten as
\begin{equation} \label{eq:rewritte_control_problem}
    \sup_{\theta\in\mathcal{A}} \E\left(R_T^\theta-\frac{\lambda}{2}[R^\theta,R^\theta]_T\right)=\sup_{\theta\in \mathcal{A}} J(\theta)+\left(V_0 -\xi_{0}^{\emptyword}  \right).
\end{equation}

We are interested in studying the concavity of this functional $J$ to deduce an existence and uniqueness result for the control problem. To this end, we define the inner product $\langle.,.\rangle_{L_2}$ by
\[\langle f,g\rangle_{L_2}:=\E\left(\int_0^T f_t g_t dt\right),\quad f,g\in\mathcal{A}. \]

We first prove that $J$ is G\^ateaux differentiable. Then, we deduce an existence and uniqueness theorem under a monotonicity condition on the G\^ateaux derivative of the functional $J$. 
\begin{lem}\label{lem: gamma_der} 
    For any $\theta,\phi,\upsilon\in \mathcal{A}$ and $\alpha_t\in T^{p}(\mathbb{R}^2)$ with $p=2\vee 2\tilde{M}\mathbb{1}_{\{\nu>0\}}$, we have that 
    \begin{align}\label{eq: gateaux_diff_P}
        \lim_{\varepsilon \to 0} \frac{1}{\varepsilon} \E\left(\upsilon_t\langle \alpha_t,\hat{\mathbb{P}}_t^{\theta+\varepsilon \phi}-\hat{\mathbb{P}}_t^{\theta}  \rangle \right)= \nu\int_0^t \E\left( \phi_s \upsilon_t\left\langle \alpha_t, \hat{\mathbb{P}}^\theta_s \otimes \word{2} \otimes \hat{\mathbb{P}}_{s,t}^\theta \right\rangle \right)ds,\quad t\leq T.
    \end{align}
\end{lem}

\begin{proof}
    The proof is given in Section \ref{sec: proofs}. 
\end{proof}

\begin{lem}\label{lem: gateau_diff} 
    For any $\theta\in \mathcal{A}$, the functional $J$ is G\^ateaux differentiable. For any $h\in \mathcal{A}$,
    \begin{align*}
        \lim_{\varepsilon\to 0} \frac{J(\theta+\varepsilon h)-J(\theta)}{\varepsilon}=\langle\nabla J(\theta), h\rangle_{L^2}, 
    \end{align*}with the G\^ateaux differential given by 
    \begin{equation*}
        \nabla J(\theta)=-2\eta \theta+\nu \mathbf{A}\theta-\lambda \sigma^2 \mathbf{B}\theta + \mu \mathbf{C}\theta, 
    \end{equation*}
    where the operators $\mathbf{A}:\mathcal{A}\rightarrow \mathcal{A}$, $\mathbf{B}:\mathcal{A}\rightarrow \mathcal{A}$ and $\mathbf{C}:\mathcal{A}\rightarrow \mathcal{A}$ are defined, for $\theta\in \mathcal{A}$, by
    \begin{align}
        &(\mathbf{A}\theta)_t =X_t^\theta- \left\langle \xi_{t} \proj{2},\hat{\mathbb{P}}^\theta_t \right\rangle+\int_t^T \E\left(\theta_s \left(1-\nu \Gamma_{t,s}^\theta \right)|\mathcal{F}_t\right)ds,\quad t\leq T, \label{eq: op_A}\\
        &(\mathbf{B}\theta)_t = \int_t^T \E\left(X_s^\theta \left(1-\nu \Gamma_{t,s}^\theta \right)|\mathcal{F}_t\right)ds+\int_t^T \E\left( \frac{1}{2} \nu \tilde{\Gamma}_{t,s}^\theta-\langle \xi_s\proj{2}, \hat{\mathbb{P}}_s^\theta\rangle|\mathcal{F}_t\right) ds,\quad t\leq T,  \label{eq: op_B}\\
        &(\mathbf{C}\theta)_t = (T-t)-\int_t^T\E(\nu  \Gamma_{t,s}^\theta|\mathcal{F}_t) ds, \quad t\leq T, \label{eq: op_C}
    \end{align}
with 
\begin{align*}
    \Gamma_{t,s}^\theta&:=\left\langle\xi_s\proj{2}, \hat{\mathbb{P}}^\theta_t \otimes \word{2} \otimes \hat{\mathbb{P}}_{t,s}^\theta \right\rangle, \quad \tilde{\Gamma}_{t,s}^\theta:= \left\langle\left(\xi_s\proj{2}\right)\shupow{2}, \hat{\mathbb{P}}^\theta_t \otimes \word{2} \otimes \hat{\mathbb{P}}_{t,s}^\theta \right\rangle,\quad t<s\leq T.
\end{align*}
\end{lem}

\begin{proof}
    The proof is given in Section \ref{sec: proofs}. 
\end{proof}

We are now ready to prove the existence and uniqueness result. To this end, we use general results of convex optimization in infinite dimensional spaces and we show that under suitable monotonicity conditions, the functional $J$ is strongly concave. We recall that, given a positive constant $\beta>0$, a map $\mathcal{T}: \mathcal{A}\to \mathbb{R}$ is $\beta-$strongly concave if, for every $\theta,\phi\in \mathcal{A}$, the following inequality holds
\begin{align} \label{eq:strongly_concavity}
    \mathcal{T}\left(\alpha \theta + (1-\alpha) \phi\right) \geq \alpha \mathcal{T}(\theta) + (1-\alpha) \mathcal{T}(\phi) +  \frac{\beta \alpha (1-\alpha)}{2} ||\theta-\phi||_{\mathcal{A}}^2,~\alpha\in[0,1],
\end{align}
with $||\theta||_{\mathcal{A}}:=\E\left(\int_0^T \theta_s^{p} ds\right)^{1/p}$, $p:=2\vee 2\tilde{M}\mathbb{1}_{\{\nu>0\}}$.  
\begin{thm}\label{thm: mono_cond_existence_unicity} For a given $\varepsilon>0$, assume that the operators $\mathbf{A},$ $\mathbf{B}$ and $\mathbf{C}$, defined by \eqref{eq: op_A}, \eqref{eq: op_B} and \eqref{eq: op_C}, satisfy the following monotonicity condition, for $\theta,\phi \in \mathcal{A},$ 
\begin{equation}\label{eq:monotonicity_cond}
\begin{aligned} 
&\bigg{ \langle} -2\eta(\theta-\phi)+\varepsilon \left(||\theta||_{\mathcal{A}}^{2-p} \theta^{p-1}-||\phi||_{\mathcal{A}}^{2-p} \phi^{p-1}\right)+\nu \left(\mathbf{A}(\theta)-\mathbf{A}(\phi)\right)-\lambda \sigma^2 \left(\mathbf{B}(\theta)-\mathbf{B}(\phi) \right)\\
&+ \mu \left(\mathbf{C}(\theta)-\mathbf{C}(\phi) \right), \theta-\phi \bigg\rangle_{L_2}\leq 0. 
\end{aligned}
\end{equation}
Then, $J$ is $\varepsilon$-strongly concave in the sense of \eqref{eq:strongly_concavity}, and there exists a unique admissible optimal control $\theta^*\in  \mathcal{A}$ satisfying  \eqref{eq:rewritte_control_problem}, which is also the unique solution of \begin{align} \label{eq: FOC}
    \nabla J(\theta^*)=0. 
\end{align}

\end{thm}

\begin{proof} 
 
    The proof of the theorem statements are consequences of well-known results of convex analysis in Banach spaces that can be found in \cite{ekeland1999convex}. First, let us prove that $J$ is $\varepsilon-$strongly concave. To this end, we define the functional $\tilde{J}$ by 
    \begin{align*}
        \tilde{J}(\theta):= -J(\theta)-\frac{\varepsilon }{2}||\theta||_{\mathcal{A}}^2. 
    \end{align*}
    If we now prove that $\tilde{J}$ is convex then, we deduce that $J$ is $\varepsilon-$strongly concave. By Lemma \ref{lem: gateau_diff}, we know that $J$ is G\^ateaux differentiable, therefore $\tilde{J}$ is also G\^ateaux differentiable. In this case,  \cite[Proposition 5.5. of Chapter I]{ekeland1999convex} gives an equivalence between the convexity of $\tilde{J}$ and the monotonicity of the G\^ateaux differential $\nabla \tilde{J}(\theta)$ defined as 
    \begin{equation*}
        \left\langle\nabla \tilde{J}(\theta),h \right \rangle_{L_2} = \lim_{\varepsilon\to 0} \frac{\tilde{J}(\theta+\varepsilon h)-\tilde{J}(\theta)} {\varepsilon},~h\in \mathcal{A}. 
    \end{equation*}
    In fact, $\tilde{J}$ is convex if, for $\theta,\phi\in \mathcal{A}$, 
    \begin{align}\label{eq: equivalence_convexity_mono}
        \left\langle \nabla\tilde{J}(\theta)-\nabla \tilde{J}(\phi),\theta-\phi \right\rangle_{L_2} \geq 0. 
    \end{align}
    Using Lemma \ref{lem: gateau_diff}, under our monotonicity assumption given by \eqref{eq:monotonicity_cond}, we deduce that \eqref{eq: equivalence_convexity_mono} is satisfied, and therefore, we obtain that $\tilde{J}$ is convex as well as $J$ is $\varepsilon-$strongly concave. As we prove that $J$ is $\varepsilon-$strongly concave, we have that $J$ is strictly concave and coercive, thus by relying on \cite[Proposition 1.2. of Chapter II]{ekeland1999convex}, we immediately have that there exists a unique admissible optimal control $\theta^*\in  \mathcal{A}$ satisfying  \eqref{eq:rewritte_control_problem}. Finally, since $J$ is  G\^ateaux differentiable, from \cite[Proposition 2.1. of Chapter II]{ekeland1999convex}, we know that the set of optimal strategies satisfying \eqref{eq:rewritte_control_problem} coincides with the set of solutions to \eqref{eq: FOC} and therefore, we deduce that the unique optimal control $\theta^* \in \mathcal{A}$, is also the unique solution of \eqref{eq: FOC}. 
\end{proof}
\begin{prop}\label{prop: example_monotonicity_ok_without_permanent}
    If $\nu=0$, then the monotonicity condition \eqref{eq:monotonicity_cond} is satisfied. 
\end{prop}

\begin{proof}
    The proof is given in Section \ref{sec: proofs}.
\end{proof}

\begin{prop}\label{prop: example_monotonicity_ok}
Let us consider that $\nu>0$ and for all $t\leq T$, $\xi_t\proj{2} \in T^{\leq 1}(\mathbb{R}^2)$. Moreover, let us assume that, for a given $0<\varepsilon<2\eta$,
\begin{equation}\label{eq:condition_time_dependent_gamma}
    \frac{1}{\nu} \left(1-\frac{2\eta-\varepsilon}{T\nu}\right)\leq  \xi_t^{\word{22}}\leq \frac{1}{\nu}, \quad t \leq T. 
\end{equation}
Then, the monotonicity condition \eqref{eq:monotonicity_cond} is satisfied.
\end{prop}

\begin{proof}
    The proof is given in Section \ref{sec: proofs}.
\end{proof}

\begin{sqremark}
The class of payoffs concerned by Proposition \ref{prop: example_monotonicity_ok} (i.e. $\xi_t\proj{2}\in T^{\leq 1}(\mathbb{R}^2)$) includes all signature payoffs with $\xi\in T^{\leq 2}(\mathbb{R}^2)$, but also some signature payoffs with $\xi\in T^{>2}(\mathbb{R}^2)$ such as, for example, Asian quadratic payoffs. 
\end{sqremark}

Theorem \ref{thm: mono_cond_existence_unicity} gives a general existence and uniqueness result, and characterizes the optimal control through the first order condition \eqref{eq: FOC}. However, this is intricate to deduce the optimal control using this condition. In the next section, we provide a more tractable verification result to derive the optimal control. 

\subsection{Verification result using infinite-dimensional Riccati equations}\label{sec: verification}
In this section, we consider a verification result for the stochastic control problem. We first give a general verification theorem such that, under given assumptions, including the existence of a solution to an infinite-dimensional system of Riccati equations, the optimal control associated with the mean-quadratic variation hedging problem can be written in a feedback form, as a linear combination of elements of the signature with time-dependent coefficients. Then, we provide two examples for which the assumptions of the verification theorem are satisfied and for which we have an explicit solution to the system of infinite-dimensional Riccati equations.  

\subsubsection{General verification theorem}
The following theorem provides a verification result, with a semi-explicit solution for the mean-quadratic variation hedging problem of signature payoffs when the trader incurs both temporary and permanent market impact, in terms of the infinite-dimensional system of Riccati equations on the extended tensor algebra space  $T((\mathbb{R}^3))$:
\begin{equation}\label{eq:infinite_dim_riccati_psi_t} 
     \begin{aligned}
        &\dot{\bpsi_t} = - \bpsi_t\proj{1} - \frac{1}{2} \sigma^2 \bpsi_t\proj{22}-\frac{\lambda}{2} \sigma^2 (\word{3}+\emptyword X_0-\tilde{\xi}_t)\shupow{2}+\frac{1}{4\eta} \bigg[ \nu (\word{3}+\emptyword X_0-\tilde{\xi}_t)-(\nu \bpsi_t\proj{2}+\bpsi_t\proj{3})\bigg]\shupow{2}\\
        &~~~~~~+\mu \left[(\word{3}+\emptyword X_0-\tilde{\xi}_t)-\bpsi_t\proj{2} \right],\\
        &\bpsi_T 
            = 0,
            \end{aligned}
    \end{equation}
     where $\tilde{\xi}_t$ is defined, for all $\mathbf{w}=\word{i_1}...\word{i_n}$ with $n\in \mathbb{N}$, by
     \begin{equation*}
          \tilde{\xi}_t\element{w}:=\begin{cases} 0,&\text{if}~\exists k\in\{1,...,n\}~\text{such that}~\word{i_k}=\word{3},\\
            \left(\xi_t\proj{2}\right)\element{w},&\text{if}~\nexists k\in \{1,...,n \}~\text{such that}~\word{i_k}=\word{3},
        \end{cases}
     \end{equation*}
     with $\xi_t$ given by \eqref{eq: def_xi_t}.
     
\begin{thm} \label{thm: verification_result}
     Let us define $\hat{Z}_t^\theta:=(t,P_t^\theta,X_t^\theta)$ and suppose that there exists $(\bpsi_t)_{t\in[0,T]}$ solution of the following infinite-dimensional system of Riccati equations \eqref{eq:infinite_dim_riccati_psi_t}. Moreover, let us assume that, for all $\theta\in \mathcal{A}$, $(\bpsi_t)_{t\in[0,T]}\in \mathcal{I}^{'}(\hat{\mathbb{Z}}^\theta)$,
     \begin{equation}\label{eq: cond_martingality}
         \E\left(\int_0^T \langle\psi_t\proj{2},\hat{\mathbb{Z}}_t^\theta \rangle^2 dt \right)<\infty,
     \end{equation}
     and that there exists $\theta^*\in \mathcal{A}$ satisfying the feedback equation
        \begin{align} \label{eq:hedgingoptimal_speed}
            \theta_t^* = \frac{1}{2\eta} \left[ \nu \left(X_t^{\theta^*}  -\left\langle \xi_t\proj{2},\hat{\mathbb{P}}^{\theta^*}_t \right\rangle \right) - \left\langle \nu \bpsi_t\proj{2}+\bpsi_t\proj{3},\hat{\mathbb{Z}}_t^{\theta^*}\right\rangle\right]. 
        \end{align}
     Then, the value of the mean-quadratic variation hedging problem is given by 
        \begin{align} \label{eq:hedgeopti}
        \sup_{\theta\in \mathcal{A}} \E\left(R_T^\theta-\frac{\lambda}{2}[R^\theta,R^\theta]_T\right) = \left(V_0 - \xi_0^{\emptyword} \right)-\bpsi_0^{\emptyword}, 
        \end{align}
    and the optimum is attained for $\theta^*$.
\end{thm}

\begin{proof}
    Assume that there exists $(\bpsi_t)_{t\in[0,T]}\in \mathcal{I}^{'}(\hat{\mathbb{Z}}^\theta)$ solution of the infinite-dimensional system of Riccati equations \eqref{eq:infinite_dim_riccati_psi_t} and define, for any $\theta\in\mathcal{A}$, the process $(U_t^\theta)_{t\in[0,T]}$ by
    \begin{equation}\label{eq:ansatz_dyn}
    U_t^\theta =  \int_0^t\left[ (\mu+\nu\theta_s)\left(X_s^\theta -\left\langle \xi_s\proj{2},\hat{\mathbb{P}}^\theta_s \right\rangle\right)- \eta \theta_s^2 - \frac{\lambda}{2}\sigma^2   \left( X_s^\theta - \left\langle \xi_s\proj{2},\hat{\mathbb{P}}^\theta_s     \right\rangle\right)^2\right]  ds -\left\langle \bpsi_t,\hat{\mathbb{Z}}^\theta_t\right\rangle,~t\leq T. 
    \end{equation}
Since, $\psi\in \mathcal{I}'(\hat{\mathbb{Z}}_t^\theta)$, an application of Ito's lemma in Theorem \ref{thm:sig_ito} leads to 
\begin{align*}
    dU_t^\theta =&\bigg{[}(\mu+ \nu\theta_t)\left(X_t^\theta -\left\langle \xi_t\proj{2},\hat{\mathbb{P}}^\theta_t \right\rangle\right)- \eta \theta_t^2 - \frac{\lambda}{2}\sigma^2   \left( X_t^\theta - \left\langle \xi_t\proj{2},\hat{\mathbb{P}}^\theta_t \right\rangle\right)^2 \bigg{]} dt - d\left\langle \bpsi_t,\hat{\mathbb{Z}}^\theta_t\right\rangle \\
    =&\bigg{[}\nu\theta_t\left(X_t^\theta -\left\langle \xi_t\proj{2},\hat{\mathbb{P}}^\theta_t \right\rangle\right)- \eta \theta_t^2 - \frac{\lambda}{2}\sigma^2   \left( X_t^\theta - \left\langle \xi_t\proj{2},\hat{\mathbb{P}}^\theta_t \right\rangle\right)^2   \\
    &-\left \langle \dot{\bpsi}_t+\bpsi_t\proj{1}+ \mu \bpsi_t\proj{2}+\theta_t (\nu\bpsi_t\proj{2}+ \bpsi_t\proj{3})+\frac{1}{2} \sigma^2 \bpsi_t\proj{22}, \hat{\mathbb{Z}}_t^\theta\right\rangle \bigg{]} dt\\
    &-\sigma \left\langle \bpsi_t\proj{2}, \hat{\mathbb{Z}}_t^\theta\right\rangle dW_t. 
\end{align*}
By a square completion, we observe that
\begin{align*}
    dU_t^\theta =& \bigg{[} -\eta \left(\theta_t-\frac{1}{2\eta} \left[ \nu \left(X_t^\theta -\left\langle \xi_t\proj{2},\hat{\mathbb{P}}^\theta_t \right\rangle \right) - \left\langle \nu \bpsi_t\proj{2}+\bpsi_t\proj{3},\hat{\mathbb{Z}}_t^\theta \right\rangle\right]  \right)^2 \\
    &+\frac{1}{4\eta} \left( \nu \left(X_t^\theta -\left\langle \xi_t\proj{2},\hat{\mathbb{P}}^\theta_t \right\rangle \right) - \left\langle \nu \bpsi_t\proj{2}+\bpsi_t\proj{3},\hat{\mathbb{Z}}_t^\theta \right\rangle\right)^2\\
    &-\frac{\lambda}{2}\sigma^2   \left( X_t^\theta - \left\langle \xi_t\proj{2},\hat{\mathbb{P}}^\theta_t \right\rangle\right)^2-\left \langle \dot{\bpsi}_t+\bpsi_t\proj{1} + \mu \bpsi_t\proj{2} +\frac{1}{2} \sigma^2 \bpsi_t\proj{22}, \hat{\mathbb{Z}}_t^\theta\right\rangle\\
    &+\mu\left(X_t^\theta -\left\langle \xi_t\proj{2},\hat{\mathbb{P}}^\theta_t \right\rangle\right)\bigg{]} dt- \sigma \left\langle \bpsi_t\proj{2}, \hat{\mathbb{Z}}_t^\theta\right\rangle dW_t. 
\end{align*}

As, for $t\leq T$, $\tilde{\xi}_t\in \mathcal{A}{(\hat{\mathbb{Z}}^\theta)}$ and $\left(\nu \psi_t\proj{2}+\psi_t\proj{3}\right) \in \mathcal{A}{(\hat{\mathbb{Z}}^\theta)}$, since $\tilde{\xi}\in \mathcal{I}'(\hat{\mathbb{Z}}^\theta)$ and $\psi\in \mathcal{I}'(\hat{\mathbb{Z}}^\theta)$, using the shuffle property of Proposition \ref{prop:shufflepropertyextended}, we deduce that 
\begin{align*}
    &\left(X_t^\theta - \left\langle \xi_t\proj{2},\hat{\mathbb{P}}^\theta_t \right\rangle\right)^2 = \left\langle \left(\mathbf{3}+\emptyword X_0-\tilde{\xi}_t\right)\shupow{2}, \hat{\mathbb{Z}}_t^\theta\right\rangle,
\end{align*}
and 
\begin{align*}
    &\left( \nu \left(X_t^\theta -\left\langle \xi_t\proj{2},\hat{\mathbb{P}}^\theta_t \right\rangle \right) - \left\langle \nu \bpsi_t\proj{2}+\bpsi_t\proj{3},\hat{\mathbb{Z}}_t^\theta \right\rangle\right)^2= \left\langle \left( \nu (\word{3}+\emptyword X_0-\tilde{\xi}_t)-(\nu \bpsi_t\proj{2}+\bpsi_t\proj{3})\right)\shupow{2}, \hat{\mathbb{Z}}_t^\theta\right\rangle. 
\end{align*}

Moreover, using the Riccati equation \eqref{eq:infinite_dim_riccati_psi_t} satisfied by $\bpsi_t$, we get that 
\begin{align*}
    dU_t^\theta = -\eta (\theta_t-\mathcal{T}_t(\theta))^2 dt -  \sigma \left\langle \bpsi_t\proj{2}, \hat{\mathbb{Z}}_t^\theta\right\rangle dW_t, 
\end{align*}
with $$\mathcal{T}_t(\theta):= \frac{1}{2\eta} \left[ \nu \left(X_t^{\theta}  -\left\langle \xi_t\proj{2},\hat{\mathbb{P}}^{\theta}_t \right\rangle \right) - \left\langle \nu \bpsi_t\proj{2}+\bpsi_t\proj{3},\hat{\mathbb{Z}}_t^{\theta}\right\rangle\right].$$

Let us now define $(M_t^\theta)_{t\in[0,T]}$ as
\begin{align*}
    M_t^\theta:=&U_t^\theta + \eta \int_0^t (\theta_s-\mathcal{T}_s(\theta))^2 ds=-\left \langle \bpsi_0, \hat{\mathbb{Z}}^\theta_0 \right \rangle - \sigma \int_0^t \left \langle \bpsi_s\proj{2}, \hat{\mathbb{Z}}^\theta_s\right \rangle dW_s. 
\end{align*}
Since we assume \eqref{eq: cond_martingality}, we have that $(M_t^\theta)_{t\in[0,T]}$ is a true martingale. In this case, since $\E(M_T^\theta|\mathcal{F}_t)=M_t^\theta$ and $\left\langle\bpsi_T,\hat{\mathbb{Z}}^\theta_T\right \rangle=0$,  we observe that 
\begin{align*}
    \left \langle \bpsi_t, \hat{\mathbb{Z}}^\theta_t \right \rangle =&-\E(M_T^\theta|\mathcal{F}_t)+\int_0^t\left[ \nu\theta_s\left(X_s^\theta -\left\langle \xi_s\proj{2},\hat{\mathbb{P}}^\theta_s \right\rangle\right)- \eta \theta_s^2 - \frac{\lambda}{2}\sigma^2   \left( X_s^\theta - \left\langle \xi_s\proj{2},\hat{\mathbb{P}}^\theta_s     \right\rangle\right)^2\right]  ds\\
    &+\eta \int_0^t (\theta_s-\mathcal{T}_s(\theta))^2 ds. \\
    =&-\E \left(\int_t^T\left[ \nu\theta_s\left(X_s^\theta -\left\langle \xi_s\proj{2},\hat{\mathbb{P}}^\theta_s \right\rangle\right)- \eta \theta_s^2 - \frac{\lambda}{2}\sigma^2   \left( X_s^\theta - \left\langle \xi_s\proj{2},\hat{\mathbb{P}}^\theta_s     \right\rangle\right)^2\right]  ds\bigg{|} \mathcal{F}_t\right)\\
    &-\eta \E \left(\int_t^T (\theta_s-\mathcal{T}_s(\theta))^2 ds\bigg{|} \mathcal{F}_t\right). 
\end{align*}
By defining $J_t(\theta)$ as 
\begin{align*}
    J_t(\theta):=\E \left(\int_t^T\left[ \nu\theta_s\left(X_s^\theta -\left\langle \xi_s\proj{2},\hat{\mathbb{P}}^\theta_s \right\rangle\right)- \eta \theta_s^2 - \frac{\lambda}{2}\sigma^2   \left( X_s^\theta - \left\langle \xi_s\proj{2},\hat{\mathbb{P}}^\theta_s     \right\rangle\right)^2\right]  ds\bigg{|} \mathcal{F}_t\right), 
\end{align*}
we obtain that 
\begin{align} \label{eq:J_t_psi_t_Z_t}
    J_t(\theta)+\left \langle \bpsi_t, \hat{\mathbb{Z}}^\theta_t \right \rangle = -\eta \E \left(\int_t^T (\theta_s-\mathcal{T}_s(\theta))^2 ds\bigg{|} \mathcal{F}_t\right). 
\end{align}
We observe that the right hand side of \eqref{eq:J_t_psi_t_Z_t} is always nonpositive and vanishes for $\theta_t=\mathcal{T}_t(\theta)$, i.e. $\theta_t = \theta_t^*$ with $\theta_t^*$ given by \eqref{eq:hedgingoptimal_speed}. As we assume that there exists an admissible control process $\theta^*\in \mathcal{A}$ such that \eqref{eq:hedgingoptimal_speed} holds, then the proof is complete. In fact, in this case, for fixed $t\leq T$, we have that, $\left \langle \bpsi_t, \hat{\mathbb{Z}}^{\theta^*}_t \right \rangle=\left \langle \bpsi_t, \hat{\mathbb{Z}}^{\theta^{'}}_t \right \rangle$, for all $\theta^{'}\in \mathcal{A}_t(\theta^*):=\left\{ \theta \in \mathcal{A}: \theta_s = \theta_s^*,~\text{for}~s\leq t\right\}$. 
Therefore, from \eqref{eq:J_t_psi_t_Z_t}, we deduce that 
\begin{equation*}
    \sup_{\theta\in\mathcal{A}_t(\theta^*)}J_t(\theta) =J_t(\theta^*) =-  \left \langle \bpsi_t, \hat{\mathbb{Z}}^{\theta^*}_t \right \rangle, 
\end{equation*}
which shows that $\theta^* $ is an optimal control. In particular, for $t=0$, we have that 
\begin{align*}
    \sup_{\theta\in\mathcal{A}}J_0(\theta) =J_0(\theta^*) = -\bpsi_0^{\emptyword},
\end{align*}
and as 
\begin{align*}
    \sup_{\theta\in \mathcal{A}} \E\left(R_T^\theta-\frac{\lambda}{2}[R^\theta,R^\theta]_T\right)  =\left(V_0 -\left\langle \xi_0,\hat{\mathbb{P}}^\theta_0 \right\rangle \right) + \sup_{\theta\in\mathcal{A}}J_0(\theta), 
\end{align*}
we finally obtain \eqref{eq:hedgeopti}. 
\end{proof}

\begin{sqremark}[Interpretation of $\theta^*$]\label{rem: interpretation_theta_opt}
$\theta^*$ has a feedback form in $\hat{\mathbb{Z}}_t^{\theta^*}$ that can be decomposed into two terms:
\begin{itemize}
    %\item (OLD) The first term $\frac{1}{2\eta} \nu \left(X_t^{\theta^*}  -\left\langle \xi_t\proj{2},\hat{\mathbb{P}}^{\theta^*}_t \right\rangle \right)$ corresponds to an arbitrage. As we assume that the P\&L is marked-to-market, when $\nu>0$, the trader will increase her P\&L at time $t<T$ by trading such as to increase her inventory if she is relatively long with respect to the perfect hedging strategy, i.e. if $X_t^\theta>\langle\xi_t\proj{2}, \hat{\mathbb{P}}^\theta_t\rangle$, or to decrease her inventory if she is relatively short ($X_t^\theta<\langle\xi_t\proj{2}, \hat{\mathbb{P}}^\theta_t\rangle$). This can be clearly seen by examining the expression of $R_T^\theta$ given by \eqref{eq: rewritte_PL}. By trading in this way, the trader, knowing that $\nu>0$, takes advantage of arbitrage to push the price in the direction that increases the value of her holdings at time $t<T$. 
    \item The first term $\frac{1}{2\eta} \nu \left(X_t^{\theta^*}  -\left\langle \xi_t\proj{2},\hat{\mathbb{P}}^{\theta^*}_t \right\rangle \right)$ corresponds to an arbitrage opportunity. Since P\&L is marked-to-market, when $\nu>0$, the trader benefits from pushing the price upward by buying more, when her inventory $X_t^\theta$ is larger than the the perfect hedging strategy $\langle\xi_t\proj{2}, \hat{\mathbb{P}}^\theta_t\rangle$. Similarly, she benefits from pushing the price downward by selling if it is smaller. This can be clearly seen by examining the expression of $R_T^\theta$ given by \eqref{eq: rewritte_PL}. When $\nu>0$, the trader can deliberately move the market to make her current position more valuable.
    
    \item The second term $-\frac{1}{2\eta} \left\langle \nu \bpsi_t\proj{2}+\bpsi_t\proj{3},\hat{\mathbb{Z}}_t^{\theta^*}\right\rangle$ expresses the trader's incentive to adjust both her inventory $X^\theta$ and, when $\nu>0$, the public traded price $P^\theta$ in the direction that increases her utility at time $t<T$. 
\end{itemize}
    
\end{sqremark}

\begin{sqremark}
    Assumptions of Theorem \ref{thm: verification_result} are intricate to prove in general, especially concerning the existence of a solution to the infinite-dimensional system of Riccati equations. Nevertheless, in the next section, we provide two concrete examples for which the existence of an explicit solution to the infinite-dimensional system of Riccati equations can be established. 
\end{sqremark}
\begin{sqremark}[Polynomial payoffs]
    If $\xi=\left(\word{2}^{\otimes m}\right)_{m=0,...M}$, then we can show that the infinite-dimensional system of Riccati equations reduces to \eqref{eq: infinite_riccati_eq_Markovian_case}. 
\end{sqremark}
\begin{sqremark}[Numerical implementation]\label{rem: num_impl_riccati}
    Care should be taken regarding the numerical resolution of the system of Riccati equations and the choice of the truncation order associated to $\psi$. In fact, the shuffle product $\shuprod$ cannot be exact since at each discretization steps of the ODE, it will double the truncation order of $\bpsi$. If we assume that $\xi_t\proj{2}\in T^{\tilde{M}}(\mathbb{R}^2)$, then, as in \cite{abi2025signature}, we decide to fix, at each step, the truncation order of $\psi$  to $2\tilde{M}$, and to consider the shuffle product projected on $T^{2\tilde{M}}(\mathbb{R}^3)$ defined as $\widetilde{\shuprod}:(T^{2\tilde{M}}(\mathbb{R}^3))^2\to T^{2\tilde{M}}(\mathbb{R}^3)$. Hence, given $\xi_t\proj{2}\in T^{\tilde{M}}(\mathbb{R}^2)$, $\psi$ is an element of $T^{2\tilde{M}}(\mathbb{R}^3)$. We refer to \cite[Section 5.1.]{abi2025signature} for a more in-depth discussion, and an analysis of the quality of convergence given different truncation orders for $\bpsi$. 
\end{sqremark}

\subsubsection{Explicit solution to the infinite-dimensional system of Riccati equations: two examples}\label{sec: explicit_expressions_riccati}

We now consider two examples for which the infinite-dimensional system of Riccati equations admits an explicit solution and for which the assumptions of the verification theorem are satisfied. The first example is the case without permanent impact i.e.~$\nu = 0$. In this case, we fall back into the \cite*{bank2017hedging} framework, so it is not surprising to obtain an existence result for the system of Riccati equations. The second example is the case where we restrict ourselves to European quadratic payoff such that $\xi=\Gamma \word{22}$ with $\Gamma\in \mathbb{R}$, and we retrieve the result of \cite{almgren2016option}. Furthermore, for both examples, we also deduce the explicit form of the optimal trading speed $\theta^*$. \\

Before discussing the two examples, we consider a lemma that gives us a sufficient condition for the existence and admissibility of $\theta^*$. 

\begin{lem}\label{prop: sufficient_cond_admissible_trading_speed}
Assume that there exists $(\bpsi_t)_{t\in[0,T]}\in \mathcal{I}^{'}(\hat{\mathbb{Z}}^\theta)$ solution of \eqref{eq:infinite_dim_riccati_psi_t} such that \eqref{eq: cond_martingality} holds. Moreover, assume there exists a constant $C>0$ such that 
     \begin{equation}\label{eq: fixed_point_assumption} 
         \bigg{|}\frac{1}{2\eta}\left\langle  \left(\nu (\word{3}+\emptyword X_0-\tilde{\xi}_t)-(\nu \bpsi_t\proj{2}+\bpsi_t\proj{3}) \right), \hat{\mathbb{Z}}_s^{\theta} \otimes \left(\nu \word{2}+\word{3} \right) \otimes \hat{\mathbb{Z}}_{s,t}^{\phi}\right\rangle \bigg{|} \leq  C,~s<t\leq T \text{ and } \theta,\phi\in \mathcal{A}.
     \end{equation} 
     Then, there exists $\theta^*\in \mathcal{A}$ satisfying \eqref{eq:hedgingoptimal_speed}.
\end{lem}
\begin{proof}
    First, for a given $\alpha>0$, we define the Banach space $\mathcal{L}^{p}(\alpha)$ by
    \begin{equation*}
    \mathcal{L}^{p}(\alpha):=\left\{\theta:[0,T]\times\Omega\rightarrow\mathbb{R}~\text{prog. measurable process such that } ||\theta||_{\mathcal{L}^{p}(\alpha)}^p:=\E\left(\int_0^T e^{-\alpha t} |\theta_t|^{p} dt\right)<\infty\right\},
    \end{equation*}
    with $p=2 \vee 2\tilde{M}\mathbb{1}_{\{\nu >0\}}$. To prove the existence of $\theta^*$, we use a fixed point approach. Namely, we prove that, for $\theta \in \mathcal{L}^{p}(\alpha)$, the functional $F$ defined such as, for $t\leq T$,
\begin{equation*}
    F(\theta)_t:= \frac{1}{2\eta}\langle   \nu (\word{3}+\emptyword X_0-\tilde{\xi}_t)-(\nu \bpsi_t\proj{2}+\bpsi_t\proj{3}), \hat{\mathbb{Z}}^\theta_t\rangle,
\end{equation*}
is a contraction. As done in the proof of Lemma \ref{lem: gamma_der}, using a variation of constant approach, we observe that, for $\theta$, $\phi\in \mathcal{L}^{p}(\alpha)$ and $t\leq T$,
\begin{align*}
    \hat{\mathbb{Z}}^{\theta}_t - \hat{\mathbb{Z}}^{\phi}_t = \int_0^t \left[ \hat{\mathbb{Z}}_s^{\theta} \otimes\left(\nu \word{2}+\word{3}\right)\otimes\hat{\mathbb{Z}}_{s,t}^{\phi}\right] (\theta_s-\phi_s)~ds.  
\end{align*}
Therefore, we have that
\begin{align*}
    |F(\theta)_t-F(\phi)_t| & = \left|\int_0^t\frac{1}{2\eta}\langle  \nu (\word{3}+\emptyword X_0-\tilde{\xi}_t)-(\nu \bpsi_t\proj{2}+\bpsi_t\proj{3}), \hat{\mathbb{Z}}_s^{\theta} \otimes\left(\nu \word{2}+\word{3}\right)\otimes\hat{\mathbb{Z}}_{s,t}^{\phi}  \rangle~(\theta_s-\phi_s)~ds\right|,
\end{align*}
and as we assume \eqref{eq: fixed_point_assumption}, using H\"{o}lder's inequality,  we observe that, for $t\leq T$,
\begin{align*}
    |F(\theta)_t-F(\phi)_t|^{p} &\leq  C^{p} t^{p-1} \int_0^t (\theta_s-\phi_s)^{p}~ds. 
\end{align*}
In this case, we infer that
\begin{align*}
    ||F(\theta)-F(\phi)||^{p}_{\mathcal{L}^{p}(\alpha)} &= \E\left(\int_0^T e^{-\alpha t} \left|F(\theta)_t- F(\phi)_t \right|^{p} dt \right)\\
    & \leq C^{p} \E\left(\int_0^T e^{-\alpha t} t^{p-1}\int_0^t (\theta_s-\phi_s)^{p}~ds dt \right) \\
    & \leq\frac{ C^{p} T^{p-1}}{\alpha}||\theta - \phi||^{p}_{\mathcal{L}^{p}(\alpha)}.
\end{align*}
 Then, there exists $\alpha>0,$ such that $F$ is a contraction for $\theta\in \mathcal{L}^{p}(\alpha)$. Therefore, using the  Banach fixed-point theorem, there exists a unique $\theta^*$  that satisfies \eqref{eq:hedgingoptimal_speed}. Since there exists $\alpha>0$ such that $\theta^*\in \mathcal{L}^{p}(\alpha),$ we deduce that
\begin{align*}
    \E\left(\int_0^T |\theta^*_t|^{p} dt\right) \leq e^{\alpha T} ||\theta^*||^{p}_{\mathcal{L}^{p}(\alpha)}  < \infty,
\end{align*}
and then $\theta^*\in \mathcal{A}$. 
\end{proof}

\begin{prop}[No permanent market impact] \label{prop: explicit_sol_without_permanent_impact} Let us assume that $\nu=0$ and $\mu=0$. Let us define the time-dependent functions $f(t)$ and $K(t,s)$ as, for $t<T$,
\begin{equation}
    f(t):=c \tanh\left(\frac{c(T-t)}{\eta}\right),
\end{equation}
with $c:=\sqrt{\frac{\lambda \sigma^2 \eta}{2}}$ and, for $t<s<T$,
\begin{equation}
    K(t,s):={\frac{c}{\eta}}\frac{\cosh\left(\frac{c(T-s)}{\eta}\right)}{\sinh\left( \frac{c(T-t)}{\eta}\right)}.
\end{equation}

Moreover, let us define $\tilde{\mathcal{E}}_t$ such that for all $\word{w}=\word{i_1}...\word{i_n}$ with $n\in \mathbb{N}$, 
\begin{equation*}
          \tilde{\mathcal{E}}_t\element{w}:=\begin{cases} 0,&\text{if}~\exists k\in\{1,...,n\}~\text{such that}~\word{i_k}=\word{3},\\
            \mathcal{E}_t\element{w},&\text{if}~\nexists k\in \{1,...,n \}~\text{such that}~\word{i_k}=\word{3},
        \end{cases}
\end{equation*}
with $\mathcal{E}_t$ given by \eqref{eq: expectation_time_augmented_sig}, and $\hat{\xi}_t$ such as, for $t<T$, 
\begin{equation}
    \hat{\xi}_t:=\int_t^T K(t,s) \tilde{\xi}_s|_{\tilde{\mathcal{E}}_{s-t}} ds.
\end{equation} 

Then, $\psi_t$ defined, for $t<T$, by
    \begin{equation}
    \begin{aligned}
        \psi_t &= f(t) \left[\word{3}+\emptyword X_0 -\hat{\xi}_t \right]\shupow{2}+\frac{\lambda \sigma^2}{2}\int_t^T \left(\tilde{\xi}_s-\hat{\xi}_s\right)\shupow{2}\bigg{|}_{\tilde{\mathcal{E}}_{s-t}} ds\\
        &+\sigma^2\int_t^T f(s) \left(\hat{\xi}_s\proj{2}\right)\shupow{2}\bigg{|}_{\tilde{\mathcal{E}}_{s-t}}ds,  
    \end{aligned}
    \end{equation}
    is a solution to the infinite-dimensional system of Riccati equation \eqref{eq:infinite_dim_riccati_psi_t} belongs to $\mathcal{I}'(\hat{\mathbb{Z}}^\theta)$, and satisfies \eqref{eq: cond_martingality}. Moreover, there exists $\theta^*\in \mathcal{A}$ solution to
    \begin{equation}
        \theta_t^*=\frac{1}{\eta} f(t) \left(\langle \hat{\xi}_t,\hat{\mathbb{Z}}_t^{\theta^*}\rangle-X_t^{\theta^*}\right),\quad t\leq T,
    \end{equation}
    and $\theta^*$ is the optimal control. 
\end{prop}

\begin{proof} 
First, let us define, for $t<s<T$, $\bell^1_{t,s}:=\left(\hat{\xi}_s\proj{2}\right)\shupow{2}\bigg{|}_{\tilde{\mathcal{E}}_{s-t}}$ and $\bell^2_{t,s}:= \left(\tilde{\xi}_s-\hat{\xi}_s\right)\shupow{2}\bigg{|}_{\tilde{\mathcal{E}}_{s-t}}$. Then, we observe that
\begin{align*}
    &\dot{\psi}_t=\dot{f}(t) \left[\word{3}+\emptyword X_0 -\hat{\xi}_t \right]\shupow{2}- 2 f(t) \left[\word{3}+\emptyword X_0 -\hat{\xi}_t \right]\shuprod \dot{\hat{\xi}}_t-\sigma^2{f}(t)\left(\hat{\xi}_t\proj{2}\right)\shupow{2}\\
    &+\sigma^2\int_t^T f(s)\dot{\bell}^1_{t,s}~ds-\frac{\lambda \sigma^2}{2}\left(\tilde{\xi}_t-\hat{\xi}_t\right)\shupow{2}+\frac{\lambda \sigma^2}{2} \int_t^T\dot{\bell}^2_{t,s}~ds\\
    &\psi_t\proj{1} =-2 f(t) \left[\word{3}+\emptyword X_0 -\hat{\xi}_t \right]\shuprod \hat{\xi}_t\proj{1}+\sigma^2 \int_t^T f(s) {\bell}^1_{t,s}\proj{1} ds+\frac{\lambda \sigma^2}{2} \int_t^T {\bell}^2_{t,s}\proj{1} ds\\
    &\psi_t\proj{22} =-2 f(t) \left[\word{3}+\emptyword X_0 -\hat{\xi}_t \right]\shuprod \hat{\xi}_t\proj{22} + 2 f(t) (\hat{\xi}_t\proj{2})\shupow{2} +\sigma^2 \int_t^T f(s) {\bell}^1_{t,s}\proj{22} ds+\frac{\lambda\sigma^2}{2} \int_t^T {\bell}^2_{t,s}\proj{22} ds   \\
    &\psi_t\proj{3} = 2 f(t) \left[\word{3}+\emptyword X_0 -\hat{\xi}_t \right],
\end{align*}
with 
\begin{equation*}
    \dot{\hat{\xi}}_t = -K(t,t) \tilde{\xi}_t + \int_t^T \left(\dot{K}(t,s) \tilde{\xi}_s|_{\tilde{\mathcal{E}}_{s-t}} + K(t,s) \dot{\tilde{\xi}}_s|_{\tilde{\mathcal{E}}_{s,t}} \right) ds. 
\end{equation*}
Using similar arguments as in the proof of Lemma \ref{lem: ode_xi_t}, we have that for a given $\bell\in T(\mathbb{R}^3)$ then, for $t<s$,
\begin{equation}
    \dot{\bell|}_{\tilde{\mathcal{E}}_{s-t}} = -\left(\dot{\bell|}_{\tilde{\mathcal{E}}_{s-t}}\right)\proj{1} - \frac{1}{2} \sigma^2 \left(\dot{\bell|}_{\tilde{\mathcal{E}}_{s-t}}\right)\proj{22}.
\end{equation}
Therefore, we observe that 
\begin{equation}\label{eq: first_part_riccati}
\begin{aligned}
     -\psi_t\proj{1}-\frac{1}{2} \sigma^2 \psi_t\proj{22}=&- 2 f(t) \left[\word{3}+\emptyword X_0 -\hat{\xi}_t \right]\shuprod \int_t^T K(t,s) \dot{\tilde{\xi}}_s|_{\tilde{\mathcal{E}}_{s,t}} ds\\
     &-\sigma^2{f}(t)\left(\hat{\xi}_t\proj{2}\right)\shupow{2}+\sigma^2 \int_t^T f(s) \dot{\bell}^1_{t,s} ds+\frac{\lambda\sigma^2}{2} \int_t^T \dot{\bell}^2_{t,s} ds. 
    \end{aligned}
\end{equation}

Moreover, we also have that 
\begin{align*}
&\dot{f}(t) = \frac{1}{\eta} f(t)^2 -\frac{\lambda}{2} \sigma^2,\quad \dot{K}(t,s) = K(t,s) K(t,t),
\end{align*}
with $K(t,t) = \frac{c^2}{{{\eta}}f(t)} = {\frac{\lambda \sigma^2}{2}} \frac{1}{f(t)}$.
Thus, we infer that 
\begin{equation}\label{eq: second_part_riccati}
\begin{aligned}
    &\dot{f}(t) \left[\word{3}+\emptyword X_0 -\hat{\xi}_t \right]\shupow{2} - 2 f(t) \left[\word{3}+\emptyword X_0 -\hat{\xi}_t \right]\shuprod \left[-K(t,t) \tilde{\xi}_t + \int_t^T \dot{K}(t,s) \tilde{\xi}_s|_{\tilde{\mathcal{E}}_{s-t}}ds\right]-\frac{\lambda \sigma^2}{2}\left(\tilde{\xi}_t-\hat{\xi}_t\right)\shupow{2}\\
    &=\dot{f}(t) \left[\word{3}+\emptyword X_0 -\hat{\xi}_t \right]\shupow{2} -\frac{\lambda}{2} \sigma^2 \left[(\word{3}+\emptyword X_0-\tilde{\xi}_t)\shupow{2}-(\word{3}+\emptyword X_0-\hat{\xi}_t)\shupow{2}\right]\\
    &=-\frac{\lambda}{2} \sigma^2 (\word{3}+\emptyword X_0-\tilde{\xi}_t)\shupow{2}+\frac{1}{4\eta} (\bpsi_t\proj{3})\shupow{2}.
\end{aligned}
\end{equation}
Putting together \eqref{eq: first_part_riccati} and \eqref{eq: second_part_riccati}, we obtain that $\psi_t$ satisfies the system of Riccati equations  \eqref{eq:infinite_dim_riccati_psi_t}. As $K(t,s)$ is bounded for $t<s<T$, $\hat{\xi}_t\in T(\mathbb{R}^3)$ for all $t<T$, then, since $f(t)$ is also bounded for all $t<T$, $\psi_t$ is well defined, has finitely many non-zero terms and belongs to $\mathcal{I}'(\hat{\mathbb{Z}}^\theta)$. We also observe that, as $\hat{\xi}_t\in T(\mathbb{R}^3)$ and $f(t)$ is bounded, there exists a positive bounded time-dependent function $C(t)$ such that, for $\theta\in \mathcal{A}$,
\begin{align*}
    \E \left(\int_0^T \langle\psi_t\proj{2}, \hat{\mathbb{Z}}_t^\theta \rangle^2 dt\right)&= \E \left(\int_0^T \langle2 f(t) [\word{3}+\emptyword X_0 -\hat{\xi}_t ]\shuprod \hat{\xi}_t\proj{2}-\sigma^2  \int_t^T f(s) {\bell}^1_{t,s}\proj{2} ds-\frac{\lambda \sigma^2}{2} \int_t^T {\bell}^2_{t,s}\proj{2} ds, \hat{\mathbb{Z}}_t^\theta \rangle^2 dt\right)\\
    &\leq \E \left(\int_0^T C(t) \left(1+|X_t^\theta|^{2} + |S_t|^{4\tilde{M}} \right)dt\right)<\infty,
\end{align*}
and hence \eqref{eq: cond_martingality} is satisfied. Furthermore, using the explicit form of $\psi_t$, the feedback equation \eqref{eq:hedgingoptimal_speed} reduces
\begin{align}\label{eq: reduced_form_feedback_ex1}
    \theta_t^* &= -\frac{f(t)}{\eta} \langle \word{3}+\emptyword X_0-\hat{\xi}_t,\hat{\mathbb{Z}}_t^{\theta^*}\rangle = \frac{f(t)}{\eta} \left(\langle \hat{\xi}_t,\hat{\mathbb{Z}}_t^{\theta^*}\rangle-X_t^{\theta^*}\right). 
\end{align}
It remains to show that there exists $\theta^*$ satisfying \eqref{eq: reduced_form_feedback_ex1} and belonging to $\mathcal{A}$. For that, it is sufficient to show that the assumption \eqref{eq: fixed_point_assumption} of Lemma \ref{prop: sufficient_cond_admissible_trading_speed} is satisfied. But as we have that 
\begin{align*}
    \psi_t\proj{3} = 2 f(t) \left[\word{3}+\emptyword X_0 -\hat{\xi}_t \right],
\end{align*}
with $f(t)$ bounded, then we observe that 
\begin{align*}
    \bigg{|}\frac{1}{2\eta}\left\langle \bpsi_t\proj{3}, \hat{\mathbb{Z}}_s^{\theta} \otimes \word{3} \otimes \hat{\mathbb{Z}}_{s,t}^{\phi}\right\rangle \bigg{|} &=\frac{1}{\eta}|f(t)|\leq \frac{1}{\eta}\sup_{t\leq T}|f(t)|,
\end{align*}
and thus the assumption \eqref{eq: fixed_point_assumption} is satisfied.
\end{proof}

\begin{prop}[European quadratic payoff]\label{prop: explicit_sol_eu_quad}  Let us assume that $\mu=0$, $\nu<\sqrt{2\eta\lambda\sigma^2}$, and $\xi = \Gamma \word{22}$ with $\Gamma\in\mathbb{R}$ such that $|\nu \Gamma|<1$. Moreover, let us consider the time-dependent function $f(t)$ given by 
\begin{equation}\label{eq: f_t_time_dep_gamma_explicit}
    f(t)=\frac{1}{2(1-\nu \Gamma)} \left[\nu+c \tanh\left(\frac{c (1-\nu \Gamma)}{2\eta}(T-t)-\frac{1}{2} \ln\left(\frac{c+\nu}{c-\nu}\right)\right)\right],
\end{equation}
with $c:=\sqrt{2\eta\lambda\sigma^2}$. Then, $\psi_t$ defined as, for $t<T$,
    \begin{equation}
        \psi_t =  f(t)\left[\word{3}+\emptyword X_0 -\tilde{\xi}_t \right]\shupow{2}+\sigma^2 \emptyword \int_t^T f(s) \Gamma^2 ds
    \end{equation}

    is a solution to the infinite-dimensional system of Riccati equation \eqref{eq:infinite_dim_riccati_psi_t}, belongs to $\mathcal{I}'(\hat{\mathbb{Z}}^\theta)$, and satisfies \eqref{eq: cond_martingality}. Moreover, there exists $\theta^*\in \mathcal{A}$ solution to
    \begin{equation}
        \theta_t^*=\frac{\nu-2f(t)(1-\nu \Gamma)}{2\eta} \left(X_t^{\theta^*}-\langle {\xi_t\proj{2}},\hat{\mathbb{P}}_t^{\theta^*}\rangle \right),
    \end{equation}
    and $\theta^*$ is the optimal control.
\end{prop}

\begin{proof}
    First, we observe that, since $\xi_t\proj{2}=\Gamma \word{2}$, then $\dot{\tilde{\xi}}_t=0$, ${\tilde{\xi}}_t\proj{2}=\Gamma \emptyword$ and ${\tilde{\xi}}_t\proj{1}={\tilde{\xi}}_t\proj{22}=0$. In this case,
    \begin{align*}
    &\dot{\psi}_t=\dot{f}(t) \left[\word{3}+\emptyword X_0 -\tilde{\xi}_t \right]\shupow{2}-\sigma^2{f}(t)\Gamma^2 \emptyword\\
    &\psi_t\proj{1} = 0 \\ 
    &\psi_t\proj{2} =-2 \Gamma f(t) \left[\word{3}+\emptyword X_0 -\tilde{\xi}_t \right]\\
    &\psi_t\proj{22} =2 f(t) \Gamma^2 \emptyword   \\
    &\psi_t\proj{3} = 2 f(t) \left[\word{3}+\emptyword X_0 -\tilde{\xi}_t \right]. 
\end{align*}
We deduce that 
\begin{equation}
    -\sigma^2{f}(t)\Gamma^2 \emptyword = -\bpsi_t\proj{1}-\frac{1}{2} \sigma^2 \bpsi_t\proj{22}.
\end{equation}
Moreover, we also observe that 
\begin{align*}
    \frac{1}{4 \eta} \left[\nu (\word{3}+\emptyword X_0-\tilde{\xi}_t)-(\nu \bpsi_t\proj{2}+\bpsi_t\proj{3})\right]\shupow{2}=\frac{1}{4 \eta} (\nu -2 f(t)(1-\Gamma\nu))^2 (\word{3}+\emptyword X_0-\tilde{\xi}_t)\shupow{2}.
\end{align*}

Thus, using the form of $f(t)$ given by \eqref{eq: f_t_time_dep_gamma_explicit}, we observe that 
\begin{equation*}
    \begin{aligned}
    &\dot{f}(t)=\frac{1}{4\eta} \left(\nu -2f(t) (1-\nu\Gamma\right))^2-\frac{\lambda}{2} \sigma^2,~t<T \text{ and } f(T)=0,
    \end{aligned}
\end{equation*}
and we infer that
\begin{equation*}
    \dot{f}(t) \left[\word{3}+\emptyword X_0-\tilde{\xi}_t\right] \shupow{2} = -\frac{\lambda}{2} \sigma^2 (\word{3}+\emptyword X_0-\tilde{\xi}_t)\shupow{2}+\frac{1}{4\eta} \left[\nu (\word{3}+\emptyword X_0-\tilde{\xi}_t)-(\nu \bpsi_t\proj{2}+\bpsi_t\proj{3})\right]\shupow{2}. 
\end{equation*}

Putting all together, we deduce that $\psi_t$ satisfies the system of Riccati equations  \eqref{eq:infinite_dim_riccati_psi_t}. Since $f(t)$ is bounded for all $t<T$, $\psi_t$ is well defined, has finitely many non-zero terms and belongs to $\mathcal{I}'(\hat{\mathbb{Z}}^\theta)$.
We also have that \eqref{eq: cond_martingality} is satisfied since, for $\theta\in \mathcal{A}$,
\begin{align*}
    \E \left(\int_0^T \langle\psi_t\proj{2}, \hat{\mathbb{Z}}_t^\theta \rangle^2 dt\right)=4 \Gamma^2 \E \left(\int_0^T f(t)^2\left(X_t^\theta-\Gamma P_t^\theta \right)^2 dt\right)<\infty. 
\end{align*}
Furthermore, using the explicit form of $\psi_t$, the feedback equation \eqref{eq:hedgingoptimal_speed} reduces
\begin{align}\label{eq: reduced_form_feedback_ex2}
    \theta_t^* &=\frac{\nu-2f(t)(1-\nu \Gamma)}{2\eta} \left(X_t^{\theta^*}-\langle \xi_t\proj{2},\hat{\mathbb{P}}_t^{\theta^*}\rangle \right).
\end{align}
It remains to show that there exists $\theta^*$ satisfying \eqref{eq: reduced_form_feedback_ex2} and belonging to $\mathcal{A}$. For that, it is sufficient to show that the assumption \eqref{eq: fixed_point_assumption} of Lemma \ref{prop: sufficient_cond_admissible_trading_speed} is satisfied. As 
\begin{align*}
    \nu (\word{3}+\emptyword X_0-\tilde{\xi}_t)-(\nu \bpsi_t\proj{2}+\bpsi_t\proj{3})=\left(\nu-2f(t)(1-\nu \Gamma)\right)\left[\word{3}+\emptyword X_0 -\hat{\xi}_t \right],
\end{align*}
with $f(t)$ bounded, $\tilde{\xi}_t\proj{2}=\Gamma \emptyword$, then we get that
\begin{align*}
    \bigg{|}\frac{1}{2\eta}\left\langle  \left(\nu (\word{3}+\emptyword X_0-\tilde{\xi}_t)-(\nu \bpsi_t\proj{2}+\bpsi_t\proj{3}) \right), \hat{\mathbb{Z}}_s^{\theta} \otimes \left(\nu \word{2}+\word{3} \right) \otimes \hat{\mathbb{Z}}_{s,t}^{\phi}\right\rangle \bigg{|} &=\frac{1}{2\eta} |\left(\nu-2f(t)(1-\nu \Gamma)\right) (1-\nu \Gamma)|\\
    &\leq \frac{1}{2\eta} \sup_{t\leq T} |\left(\nu-2f(t)(1-\nu \Gamma)\right) (1-\nu \Gamma)|,
\end{align*}
and the assumption \eqref{eq: fixed_point_assumption} is satisfied. 
\end{proof}

\section{Numerical illustration} \label{sec: numerical_results}
In this section, we illustrate numerically  the signature-based hedging strategies given by Theorem~\ref{thm: verification_result}  by solving numerically the system of infinite-dimensional Riccati equations \eqref{eq:infinite_dim_riccati_psi_t}. First we consider exact signature payoffs such as European or Asian quadratic options and then we take more general path-dependent payoffs that we approximate by signature payoffs relying on the universal approximation theorem.\\

Moreover, we  compare the signature-based strategy taking into account both temporary and permanent market impact with the optimal hedging strategy that only accounts for temporary market impact  deduced for general path-dependent payoffs in \cite*{bank2017hedging}.\\

For our numerical results, we set the initial inventory $X_0$  equal to the Bachelier $\Delta$. For signature payoffs, this reduces to $X_0 = \left(\xi_0\proj{2} \right)^{\emptyword}$. Moreover, we assume that the initial portfolio value $V_0$ is equal to this indifference price $\pi$ defined in Section \ref{sec: MQV_hedging_impact}. Therefore, the initial portfolio value will vary depending on the choice of market impact parameters $\eta$ and $\nu$, which will impact the final value of the P\&L but not the trading speed. It reduces to $V_0 = \xi_0^{\emptyword}-\psi_0$ for signature payoffs. Finally, recall that, as explained in Remark \ref{rem: num_impl_riccati}, when we solve numerically the Riccati equation \eqref{eq:infinite_dim_riccati_psi_t}, we consider the shuffle product projected on $T^{2\tilde{M}}(\mathbb{R}^3)$ defined as $\widetilde{\shuprod}:(T^{2\tilde{M}}(\mathbb{R}^3))^2\to T^{2\tilde{M}}(\mathbb{R}^3)$, such that $\bpsi_t\in T^{2\tilde{M}}(\mathbb{R}^3)$, $\tilde{M}$ being the truncation order $\xi_t\proj{2}\in T^{\tilde{M}}(\mathbb{R}^2)$. Unless otherwise stated, we consider the following parameters for our numerical results: $$S_0=10,~\mu=0,~\sigma=2,~T=0.2,~\eta=0.001,~\nu=0.001,~ \lambda=0.01.$$

\subsection{Signature payoffs} 
We consider two particular signature payoffs:
\begin{itemize}
    \item European quadratic payoff: $H_T^\theta=N \times\ (P_T^\theta-K)^2$,
    \item Asian quadratic payoff: $H_T^\theta=N \times \left(\frac{1}{T}\int_0^T P_t^\theta dt-K\right)^2$,
\end{itemize}
where $N$ corresponds to a given nominal.\\

First, since the European quadratic payoff has a constant Gamma, an application of Proposition~\ref{prop: explicit_sol_eu_quad}  yields an explicit expression for the optimal trading speed, in the spirit of \cite{almgren2016option}. Thus, this allows us to perform a sanity check on the stability of the numerical solution of the Riccati equation \eqref{eq:infinite_dim_riccati_psi_t}.  
Figure \ref{fig:compa_eu_quadratic_options}  compares a sample path of trading speed obtained for a given European quadratic option. As expected, we see that the optimal strategy obtained by numerically solving the Riccati equation perfectly matches the closed form expression deduced in Proposition \ref{prop: explicit_sol_eu_quad}. This suggests that the numerical resolution of the equation appears to be stable.

\begin{figure}[H]
\begin{centering}
\includegraphics[width=0.47\textwidth]{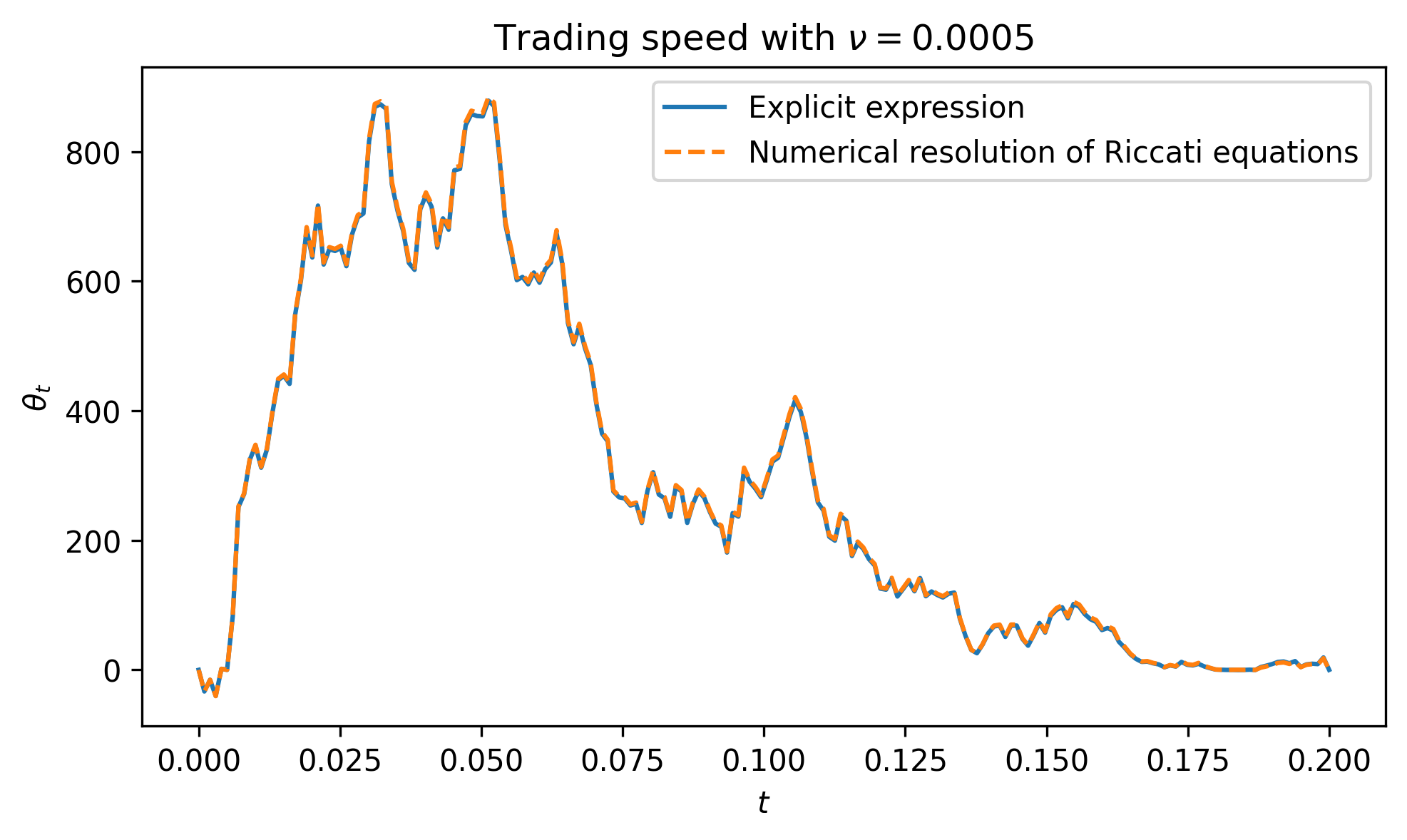} \includegraphics[width=0.47\textwidth]{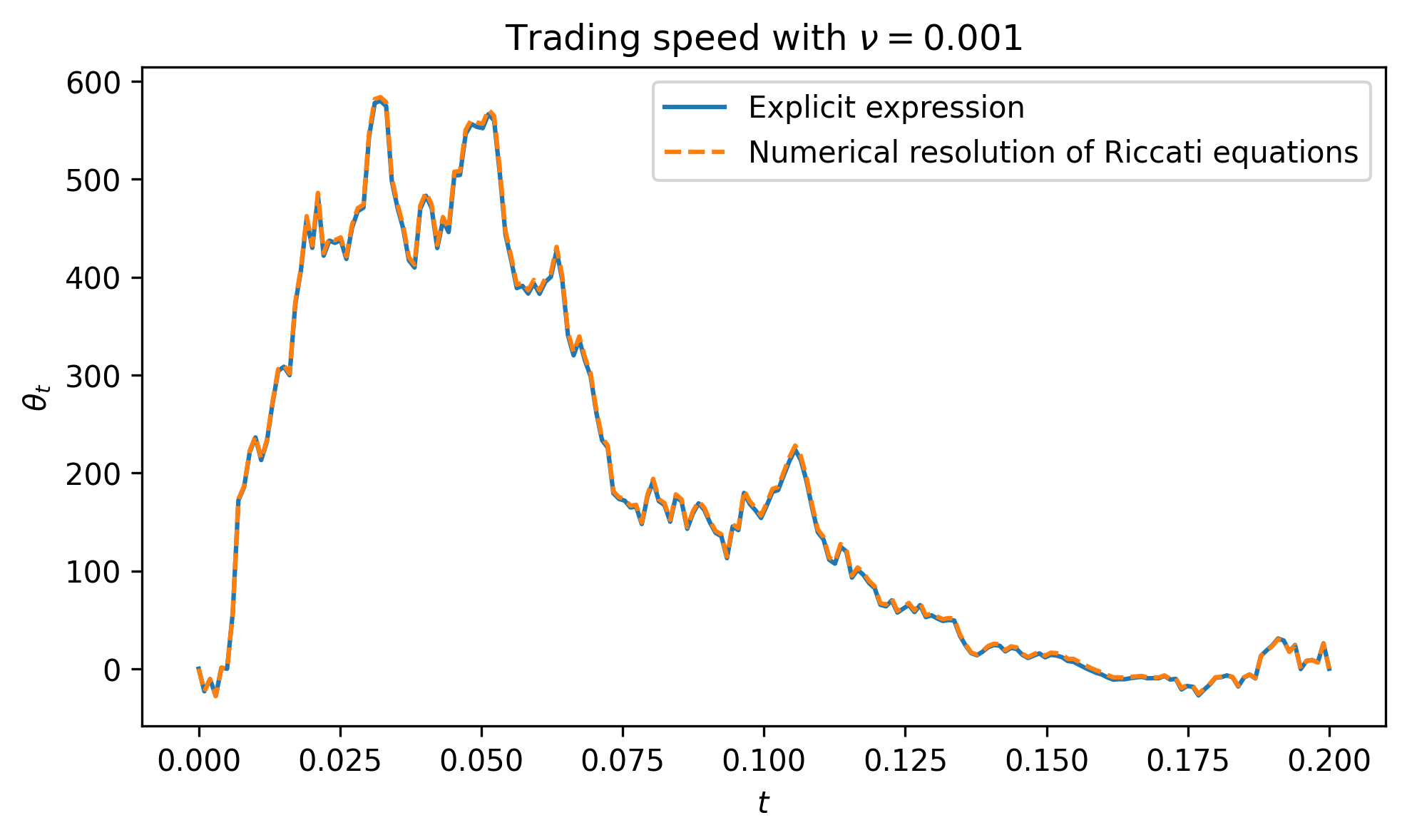}
\par\end{centering}
\caption{\protect\label{fig:compa_eu_quadratic_options} Sample path of trading speed for the European quadratic payoff: explicit expression given by Proposition \ref{prop: explicit_sol_eu_quad} vs numerical resolution of Riccati equation with truncated shuffle product. Parameters: $M=2$, $K=S_0$ and $N=200$.}
\end{figure}

Moving to Asian quadratic options, we no  longer have an explicit expression for the optimal trading speed, except for $\nu=0$, as stated by Proposition \ref{prop: explicit_sol_without_permanent_impact}, where we retrieve the same strategy as in \cite*{bank2017hedging}. Therefore, for $\nu>0$,  numerically solving the Riccati equation becomes essential to determine the optimal strategy.
As a sanity check, as $\nu \to 0$, the solution obtained from the Riccati system is expected to converge to the explicit solution. This is illustrated in Figure~\ref{fig:compa_asian_quadratic_option}, which shows a sample path of the optimal trading speed and inventory for a given Asian quadratic option under different permanent market impact parameters. We clearly observe the convergence phenomenon as $\nu \to 0$, providing an additional validation of our numerical implementation of the Riccati equations. When permanent market impact is introduced ($\nu > 0$), a gap appears whose magnitude increases with the value of the permanent impact parameter $\nu$. As expected, the larger $\nu$ is, the slower and flatter the resulting trading speed becomes.\\

\begin{figure}[H]
\begin{centering}
\includegraphics[width=1\textwidth]{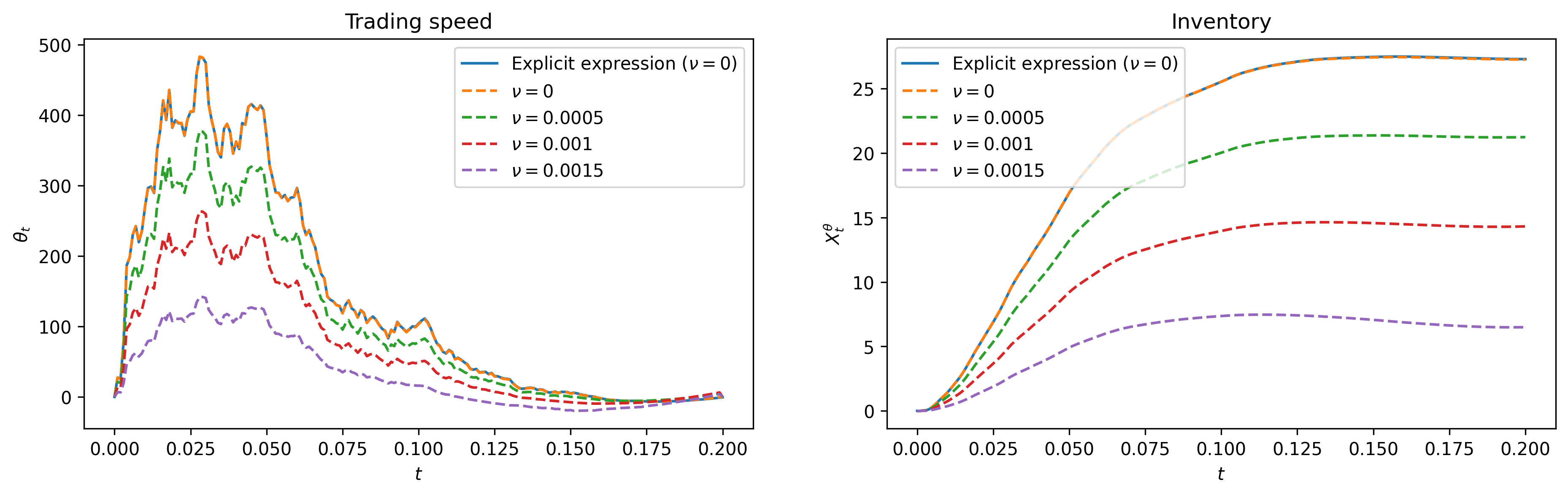}
\par\end{centering}
\caption{\protect\label{fig:compa_asian_quadratic_option}Sample path of trading speed and inventory for Asian quadratic payoff. Solid line:  explicit hedging strategy assuming $\nu=0$ given by Proposition \ref{prop: explicit_sol_without_permanent_impact}; dashed lines: signature-based strategies obtained by numerically solving the Riccati equations with projected shuffle product. Parameters: $M=4$, $K=S_0$, and $N=200$.}
\end{figure}

Moreover, to highlight the benefits of incorporating the permanent market impact information into the trading strategy, we compare the P\&L generated by two traders, when $\nu>0$. The first one is aware of the permanent market impact, incorporates this information into her strategy and follows the optimal trading speed for $\nu>0$. The second one ignores the permanent impact and follows the explicit trading strategy assuming $\nu=0$ given by Proposition \ref{prop: explicit_sol_without_permanent_impact}.  We assume that both traders have the same risk aversion and ask an indifference price to hedge the option.   Therefore, the trader assuming $\nu>0$ asks a price $\pi(\nu>0)$ while the second trader, ignoring the permanent impact, asks $\pi(\nu=0)$. Both traders' strategies are evaluated using the same criterion with $\nu > 0$: the first trader's strategy is optimal, whereas the second trader's is not. Results are shown in Figure \ref{fig:hist_asian_quadratic_options}. We can see that, when there is a permanent market impact, there is an advantage to be gained by incorporating this information into the hedging strategy and following the signature-based strategy assuming $\nu>0$, rather than the strategy assuming only a temporary impact. Obviously, this information gain depends on the size of the permanent impact parameter $\nu$ and will be higher when the permanent impact is higher. The gap between the two P\&Ls can be attributed to two things:
\begin{itemize}
    \item the trader omitting the permanent impact will misprice the option, underestimates the risk she takes (see Table \ref{tab: indif_price_asian_quad}). This will affect the final value of her P\&L,
    \item the trader who overlooks the permanent impact will trade too aggressively (see Figure \ref{fig:compa_asian_quadratic_option}), and push the price further in the wrong direction, which will adversely impact her P\&L over time. The optimal strategy with $\nu>0$, taking into account this impact, attempts to push the price in the right direction (see Remark \ref{rem: interpretation_theta_opt}).  
\end{itemize}

\begin{figure}[H]
\begin{centering}
\includegraphics[width=0.5\textwidth]{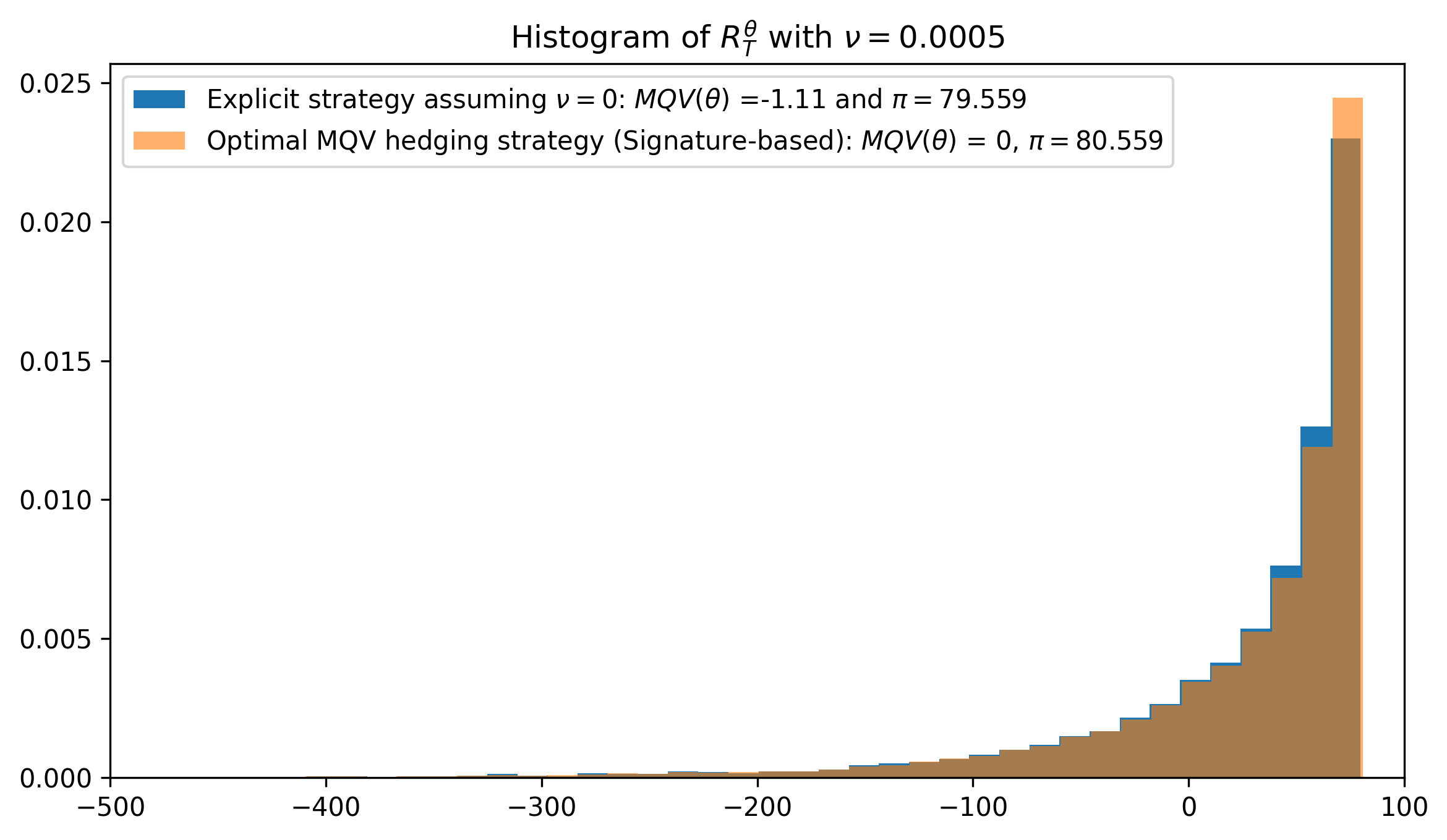}\includegraphics[width=0.5\textwidth]{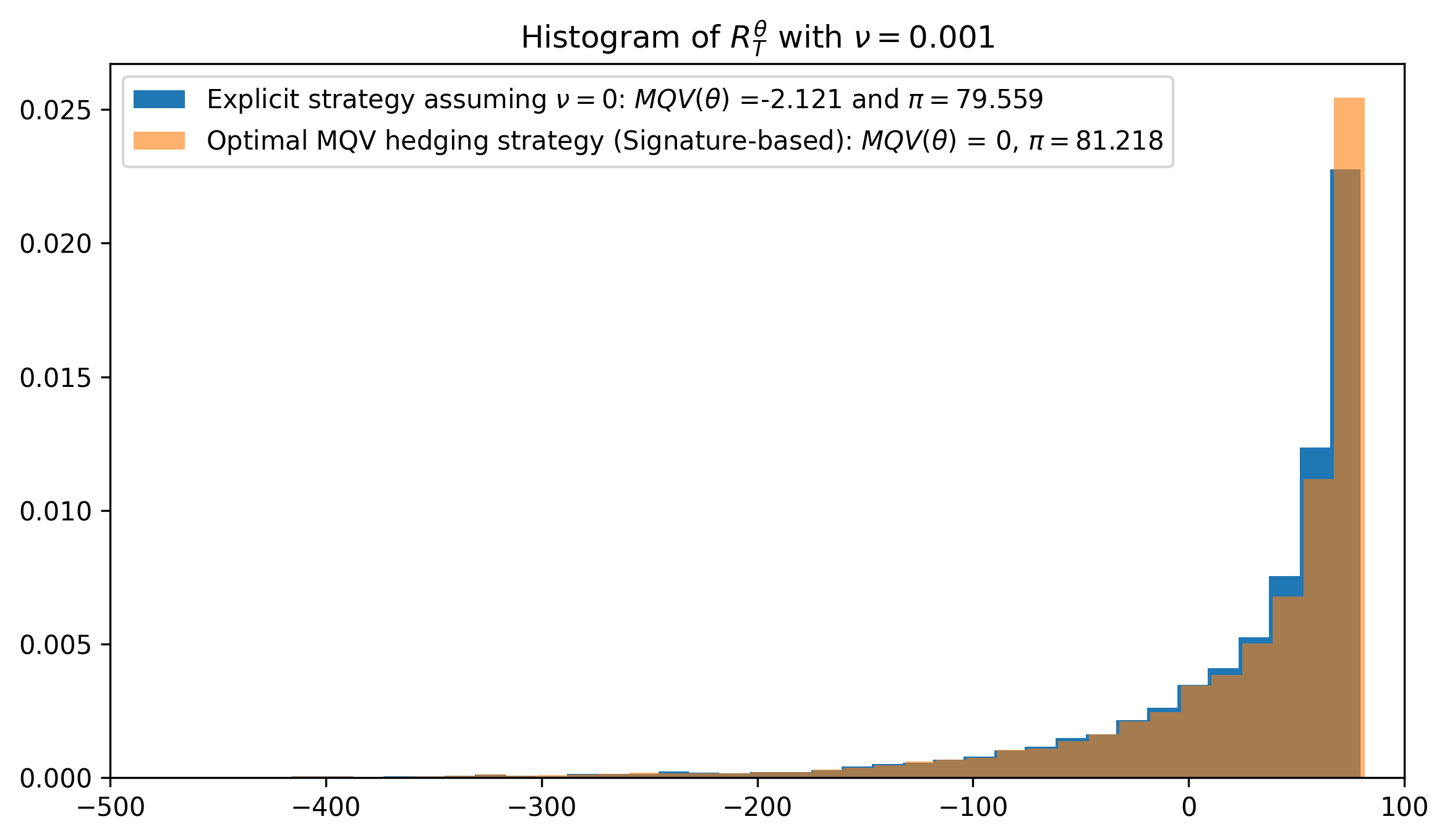} \\
\includegraphics[width=0.5\textwidth]{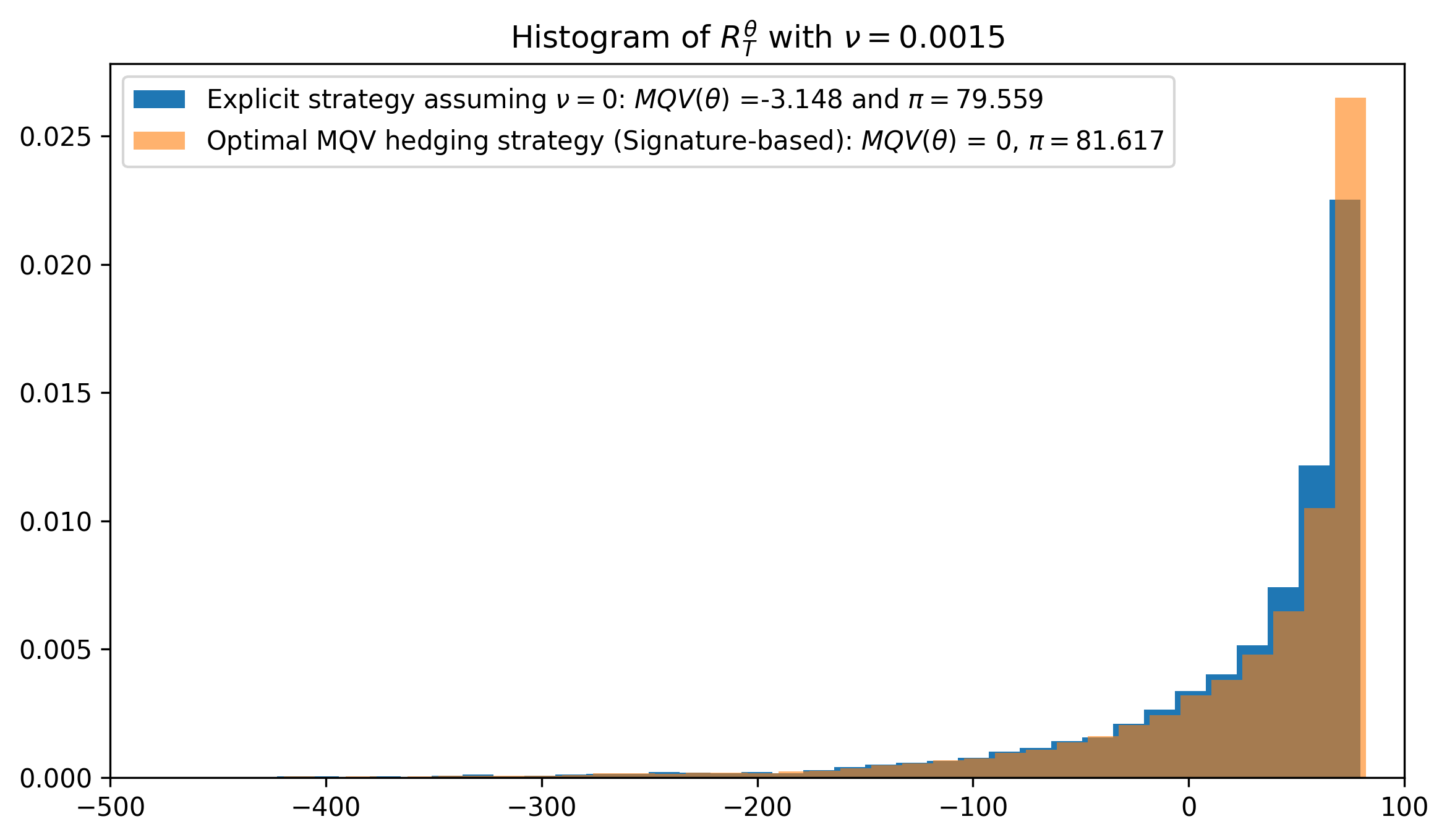} 
\par\end{centering}
\caption{\protect\label{fig:hist_asian_quadratic_options}Histogram of $R_T^\theta$ for Asian quadratic payoff with $\nu>0$: signature-based strategy vs explicit hedging strategy assuming $\nu=0$ given by Proposition \ref{prop: explicit_sol_without_permanent_impact}. Other parameters: $M=4$, $K=S_0$, and $N=200$. $MQV(\theta)=\E\left(R_T^\theta-\frac{\lambda}{2}[R^\theta,R^\theta]_T\right)$.}
\end{figure}

\begin{table}[H]
\centering
\begin{tabular}{|c|c|c|c|c|}
\hline
Bachelier price & $\nu=0$ & $\nu=0.0005$ & $\nu=0.001$& $\nu=0.0015$ \\ \hline
53.333 & 79.559 & 80.559 & 81.218 & 81.617 \\ \hline
\end{tabular}
\caption{\protect\label{tab: indif_price_asian_quad}Indifference prices $\pi$ for Asian quadratic payoff with respect to $\nu$. Other parameters: $M=4$, $K=S_0$, and $N=200$.}
\end{table}

\subsection{Non-signature path-dependent payoffs}\label{sec: hedging_call}
We consider now payoffs that cannot be exactly represented as signature payoffs as it is the case for European, Asian or barrier call options. By relying on the universal linearization property of signatures, see \cite*{lyons2020non, cuchiero2025universal}, we can approximate those path-dependent payoffs by signature payoffs, i.e. by linear combination of time-augmented signature elements. More precisely, for a continuous functional $F$, a continuous semimartingale $(Y_t)_{t\in[0,T]}$ valued in $\mathbb{R}^d$, and a given truncation order $M>0$, there exists $\bell\in T^M(\mathbb{R^{d+1}})$ such that
\begin{equation*}
    F\left((t,Y_t)_{t \leq T}\right)\approx \langle \bell, \hat{\mathbb{Y}}_t \rangle.
\end{equation*}

This means that the semi-explicit signature-based hedging strategy can also be used to hedge general (non-signature) path-dependent payoffs. However, to apply the signature-based hedging strategy, we first have to regress those path-dependent payoffs against truncated signature payoffs to determine $ \bell$, and then solve the infinite-dimensional Riccati equation associated with the signature approximate payoff $\langle \bell, \hat{\mathbb{Y}}_t\rangle $. Algorithm \ref{alg:regression} explicitly states how to get a linear signature representation for a given path-dependent payoff. Furthermore, as we consider non-signature payoffs, the indifference price $\pi$ does not have an explicit form and should also be approximated. In this sense, we approximate the indifference price by $\tilde{\pi}$ such that it satisfies the following equation
\begin{equation}\label{eq: sig_approx_indif_price}
    \E \left(R_T^{\theta^{*}_{sig}} - \frac{\lambda}{2} [R^{\theta^{*}_{sig}},R^{\theta^{*}_{sig}}]_T\right) = 0,
\end{equation}
where 
\begin{equation*}
    R_t^{\theta^{*}_{sig}}=(\tilde{\pi}-X_0S_0)+X_t^{\theta^{*}_{sig}} P_t^{\theta^{*}_{sig}}-\int_0^t \tilde{P}^{\theta^{*}_{sig}}_s {(\theta^{*}_{sig})}_s ds -h\left(\left(s,P_s^{\theta^*_{sig}}\right)_{s\leq t}\right),~t\leq T,
\end{equation*}
with $h\left(\left(s,P_s^{\theta^*_{sig}}\right)_{s\leq t}\right)$ the Bachelier price of the path-dependent payoff at time $t<T$, and $\theta^*_{sig}$ the optimal trading speed associated to the regressed signature payoff.

\begin{algorithm}[H]
        \caption{Regression of path-dependent payoffs against truncated signature payoffs}\label{alg:regression} 
        Consider a payoff of the following form
        $$ H_T= F\left((t,Y_t)_{t \leq T}\right), $$
        with $(Y_t)_{t\in [0,T]}$ a given stochastic process valued in $\mathbb{R}$. 
        
        \textbf{Input}:
        \begin{itemize}
            \item Fix $J > 0$ the number of points in the discretization of $[0, T]$,
            \item Fix $L > 0$ the number of realisations of trajectory of $Y$,
            \item Fix $M \geq 0$ the truncation order of $\bell$,
        \end{itemize}
        
        \textbf{Online}:
        \begin{enumerate}
            \item Generate $L$ realizations of the $Y$, denoted by $Y^{(1)}, \ldots, Y^{(L)}$,
            \item For each realization $l = 1, \ldots, L$, compute $H_T^{(l)}$ and the truncated signature $\hat{\mathbb{Y}}_T^{(l),~\leq M}$ up to order $M$ for $t \in [0, T]$,
            \item
            Regress $(H_T^{(l)})_{1 \leq l \leq L}$ against $(\hat{\mathbb{Y}}_T^{(l),~\leq M})_{1 \leq l \leq L}$ to learn the coefficients of $\bell \in T^M(\mathbb{R}^2)$ that minimize
            $$ \mathcal{L} = \frac{1}{L}\sum_{l=1}^L \left| H^{(l)}_{T} - \left \langle \bell, \hat{\mathbb{Y}}_T^{(l),~\leq M}\right \rangle \right|^2. $$
        \end{enumerate}
    \end{algorithm}

Let us move on to some numerical illustrations. We consider different types of payoffs: 
\begin{itemize}
    \item European call: $H_T^\theta=N \times(P_T^\theta-K)_+$,
    \item Asian call: $H_T^\theta=N \times (\frac{1}{T} \int_0^TP_s^\theta ds-K)_+$,
    \item One-touch max: $H_T^\theta=N \times \mathbb{1}_{\{\max_{t<T} P_t^\theta \geq H\}}$,
    \item Look-back call with floating strike: $H_T^\theta=N \times \left(P_T^\theta-\min_{\{t<T\}} P_t^\theta\right)$,
\end{itemize}
where $N$ corresponds to a given nominal. \\ \\
First, recall that, for general European options, \cite{almgren2016option} characterized the optimal hedging strategy as a solution to the HJB equation given by \eqref{eq: HJB_european_payoff}. Using this information, for an European call payoff, we compare the signature-based strategy associated to the regressed signature payoff with the optimal strategy obtained by numerically solving the corresponding HJB equation given by \eqref{eq: HJB_european_payoff}. This allows us to see whether, in this context of market impact, the signature-based strategies are able to well approximate the optimal strategy when payoffs do not admit an exact signature representation.\\

We consider two types of market: a frictionless market with $\eta=\nu=0$ and a frictional market with $\eta>0$ and $\nu>0$. 
In the frictionless market, see Figure \ref{fig:compa_eu_call_options_frictionless}, the signature-based strategy provides a close approximation of the perfect hedging strategy for small times $t$, but it fails to recover the correct behavior as $t\to T$. The discrepancy is primarily due to the strong nonlinearity of the perfect hedging strategy for call options near maturity. {This observation mitigates the conclusion of \cite*[Section 7.1]{lyons2020non} regarding the ability of signature-based strategies to accurately recover the perfect hedging strategy in a frictionless setting}. By contrast, in the presence of market frictions, see  Figure \ref{fig:compa_eu_call_options}, the signature approximation performs significantly better. In this case, the sample paths of trading speeds and inventories generated by the signature-based approach are much closer to those of the optimal strategy. This improvement can be attributed to two effects: first, market impact naturally smoothens the optimal trading speeds; second, the signature-approximated payoff leads to a smoother (target) perfect hedging strategy compared to the call payoff. Together, these effects mitigate approximation errors and enhance the accuracy of the signature-based strategy. Thus, in the frictional market, the signature-based strategy can be considered a sufficiently accurate approximation of the optimal strategy, showing a striking improvement compared to the frictionless setting (compare the error on the inventory in Figures \ref{fig:compa_eu_call_options_frictionless} and \ref{fig:compa_eu_call_options}). Clearly, in the absence of frictions, approximating the payoff of a call option by a fifth-order polynomial is not sufficiently accurate for pricing and hedging. However, once market frictions are introduced, the fifth-order polynomial approximation becomes significantly more accurate, even though it additionally requires  truncating and numerically solving infinite-dimensional Riccati equations.

%It is worth noting that numerical instabilities may arise when solving the HJB equation, which can cause the computed optimal strategy to deviate from the true optimal strategy. Nevertheless, the primary objective of this comparison is to assess the ability of signature-based strategies to approximate non-signature strategies.

\begin{figure}[H]
\begin{centering}
\includegraphics[width=1\textwidth]{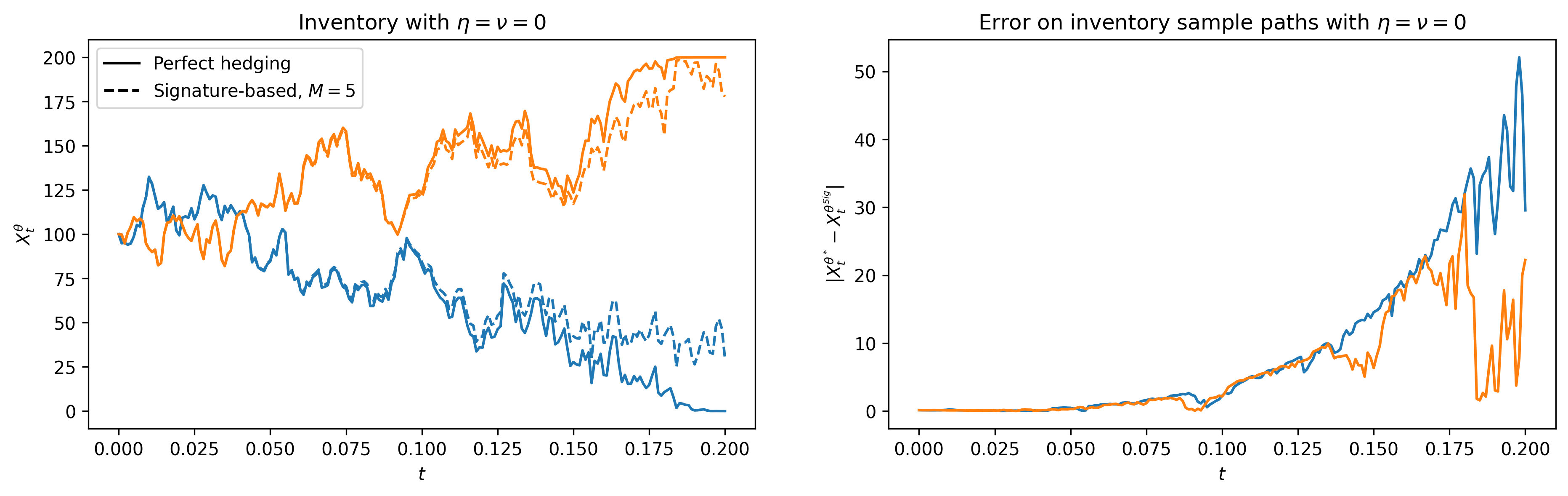}
\par\end{centering}
\caption{\protect\label{fig:compa_eu_call_options_frictionless}Sample paths of inventory $X_t^\theta$ for European call payoff with $\eta=\nu=0$: signature-based approximated strategy vs Bachelier Delta hedging strategy. Parameters: $M=5$, $K=S_0$, and $N=200$.}
\end{figure}

\begin{figure}[H]
\begin{centering}
\includegraphics[width=1\textwidth]{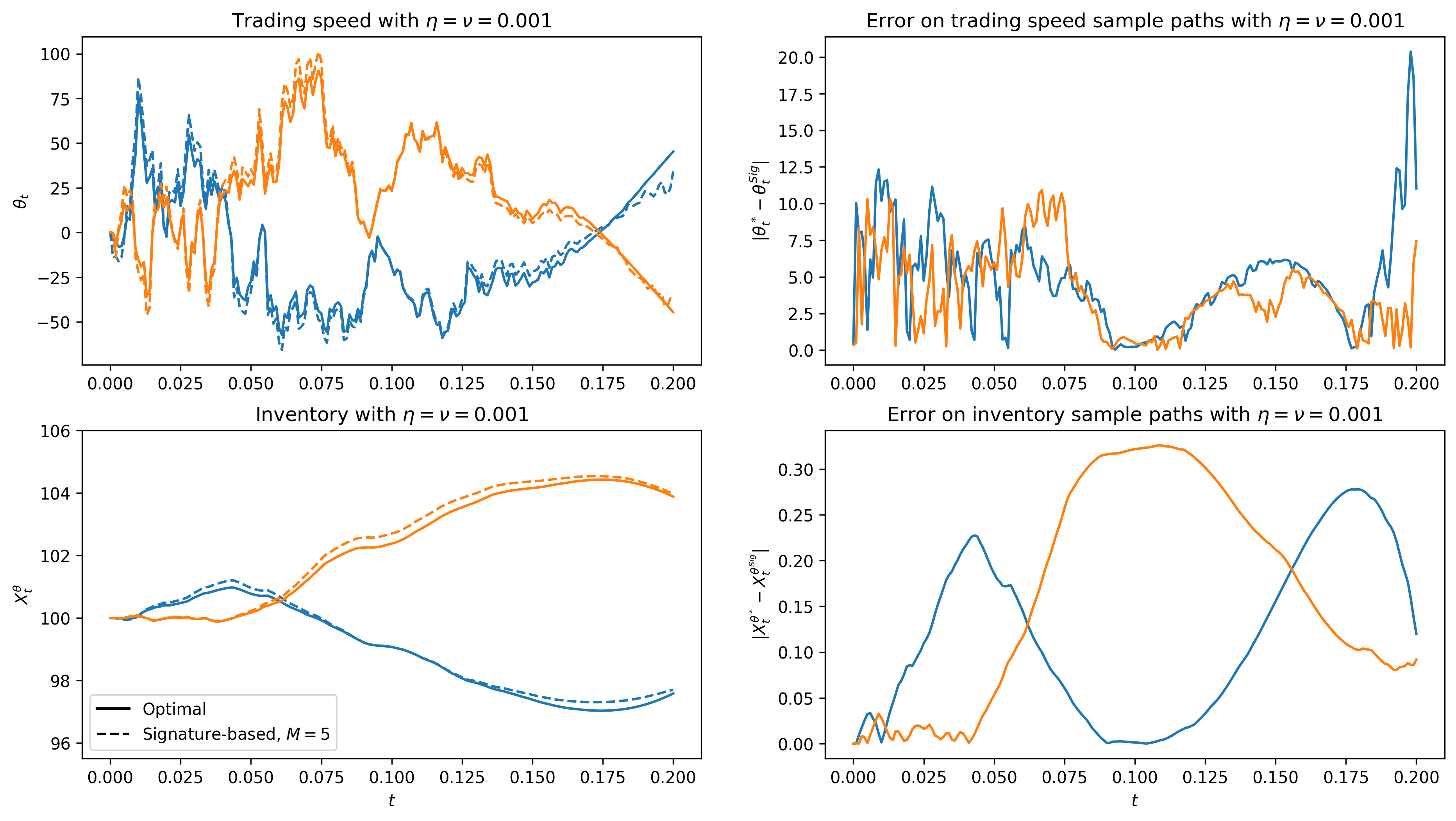}
\par\end{centering}
\caption{\protect\label{fig:compa_eu_call_options}Sample paths of trading speed $\theta_t$ and inventory $X_t^\theta$ for European call payoff with $\eta=\nu=0.001$: signature-based approximated strategy vs optimal hedging strategy by solving the HJB equation \eqref{eq: HJB_european_payoff}. Parameters: $M=5$, $K=S_0$, and $N=200$.}
\end{figure}

We now emphasize the relevance of signature-based strategies to hedge non-signature path-dependent payoffs when $\nu>0$. To this end, we consider different types of options: European call, Asian call, one-touch and look-back call with floating strike. As in the previous subsection, we compare the P\&L generated by two types of traders: the first one is aware of the permanent market impact and follows the signature-based strategy for $\nu>0$, while the second one ignores the permanent impact and  follows the trading strategy of \cite*{bank2017hedging}, which consists of tracking the Bachelier Delta of the non-signature payoffs. Both traders' strategies are evaluated using the same criterion with $\nu > 0$.\\

Figure \ref{fig:hist_call} shows the histograms of the P\&L at maturity generated by two strategies for different types of payoffs. Looking at the histograms, we remark that even though the signature-based strategy is an approximation of the optimal strategy, it seems very relevant for hedging non-signature payoffs in the presence of market frictions. Furthermore, if we compare the mean-quadratic variation functional $MQV(\theta)$, we observe that the signature-based strategy outperforms the \cite*{bank2017hedging} strategy that does not take into account the permanent impact. Thus, we see that there is therefore a substantial advantage in favoring the signature-based hedging strategy. The discrepancy between the P\&Ls can be attributed to the same factors as in the case of Asian quadratic payoff. At inception, the trader that follows \cite*{bank2017hedging} strategy misprices the options regarding the risk she takes when $\nu>0$ (see Table \ref{tab: indif_price_call_payoffs}). Moreover, during the trading period, the trader that does not take into account the permanent impact will trade too aggressively, and this will more adversely impact her P\&L over time.

\begin{figure}[H]
\begin{centering}
\includegraphics[width=0.5\textwidth]{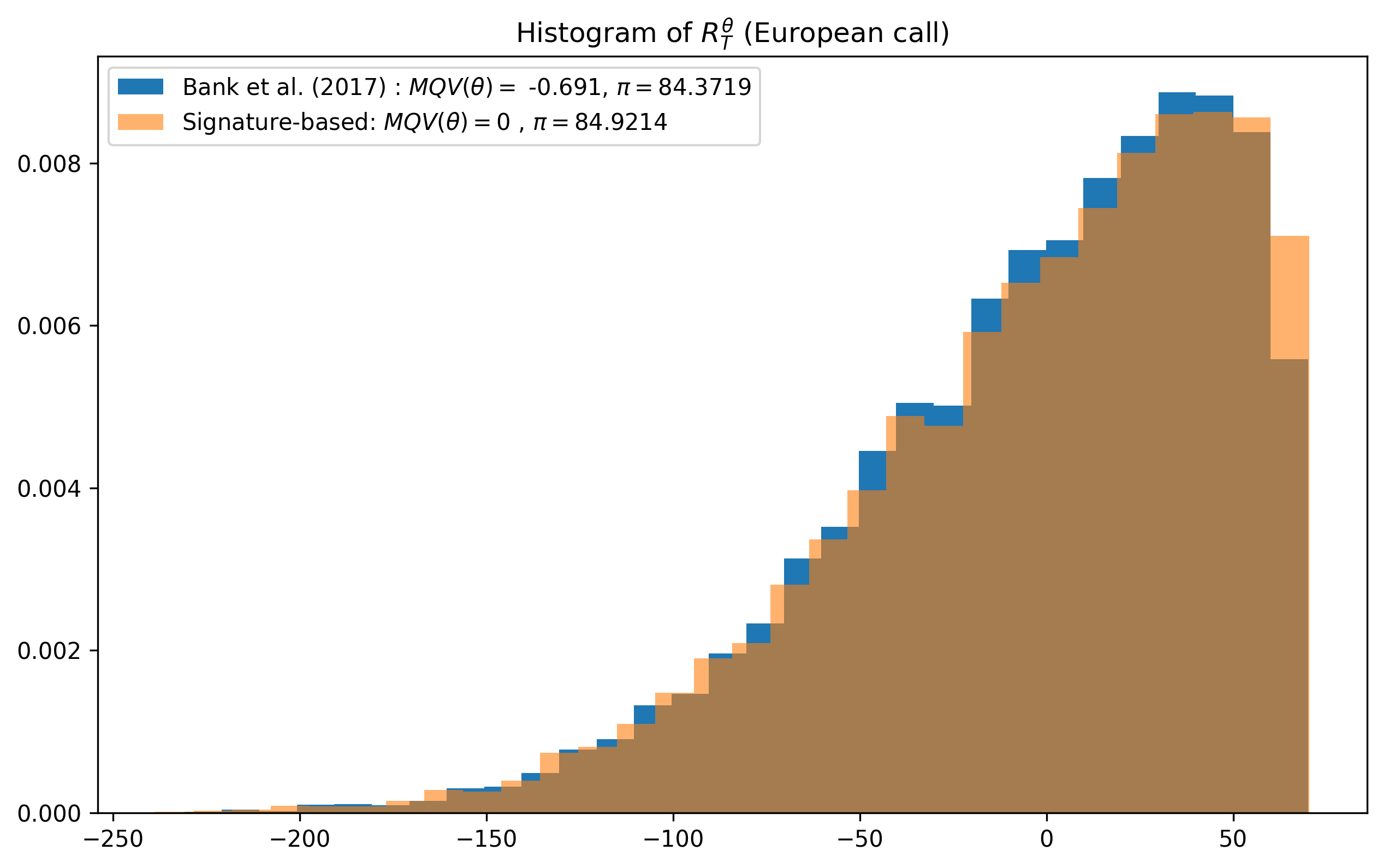}\includegraphics[width=0.5\textwidth]{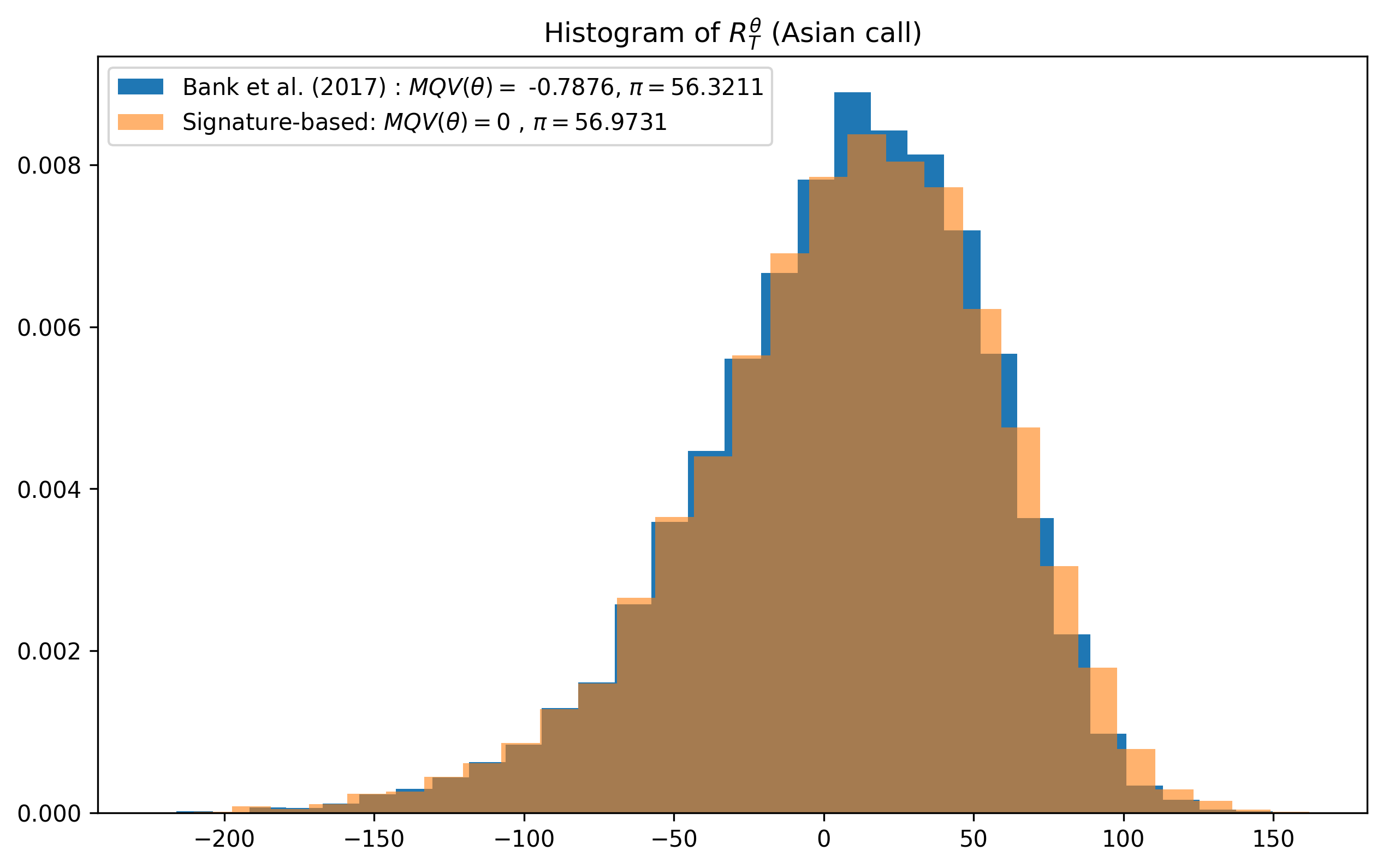} \\
\includegraphics[width=0.505\textwidth]{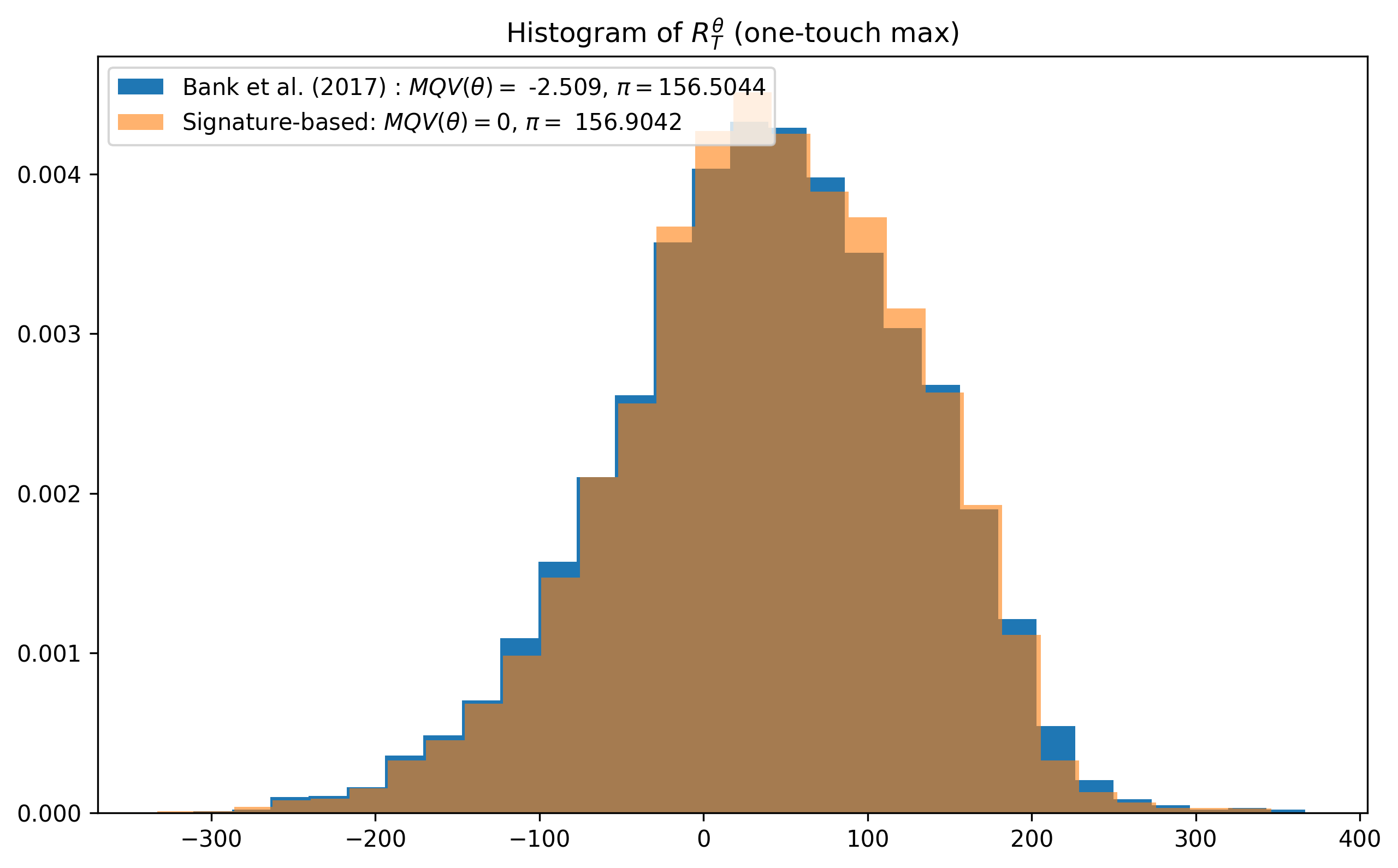}\includegraphics[width=0.5\textwidth]{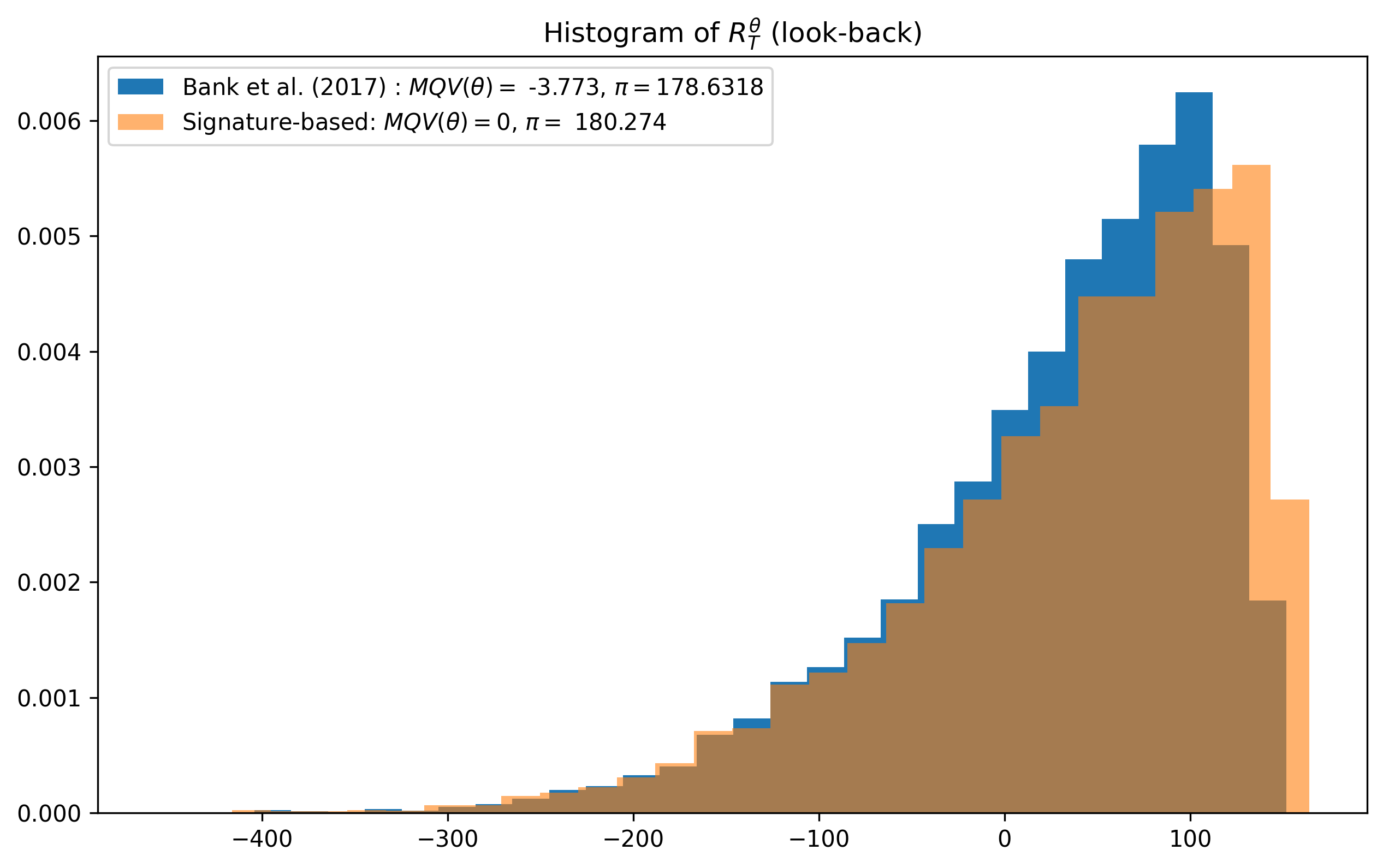}
\par\end{centering}
\caption{\protect\label{fig:hist_call}Histogram of $R_T^\theta$ for different payoffs: signature-based hedging strategy vs \cite*{bank2017hedging} hedging strategy ($\nu=0)$. Parameters: $M=5$, $K=S_0$, $H=1.05\times S_0$ and $N=200$. $MQV(\theta)=\E\left(R_T^\theta-\frac{\lambda}{2}[R^\theta,R^\theta]_T\right)$.}
\end{figure}

\begin{table}[H]
\centering
\begin{tabular}{|c|c|c|c|}
\hline
Payoff &Bachelier price & \cite*{bank2017hedging} ($\nu=0$) & Sig-approx. ($\nu=0.001$) \\ \hline
European call & 71.3649 & 84.3719  & 84.9214 \\ \hline
Asian call &41.2025 &56.3211  & 56.9731 \\ \hline
One-touch  &115.2300 & 156.5044 & 156.9042 \\ \hline
Look-back  &142.7299 & 178.6318  &180.2740   \\ \hline
\end{tabular}
\caption{\protect\label{tab: indif_price_call_payoffs}Indifference prices for different payoffs: \cite*{bank2017hedging} indifference price is obtained by using the \cite*{bank2017hedging} strategy with Bachelier Delta as tracking strategy; the indifference price for $\nu=0.001$ is approximated by $\tilde{\pi}$ using \eqref{eq: sig_approx_indif_price}. Parameters: $M=5$, $K=S_0$, $H=1.05\times S_0$, and $N=200$.}
\end{table}

\section{Proofs}\label{sec: proofs}

\subsection{Proof of Lemma \ref{lem: gamma_der}}
\begin{proof}
First, for given $\theta,\phi\in \mathcal{A}$, we have that 
\begin{align*}
    &d\hat{\mathbb{P}}^{\theta}_t = \hat{\mathbb{P}}_t^{\theta} \otimes d\hat{P}_t^\theta,\quad d\hat{\mathbb{P}}^{\theta+\varepsilon \phi}_t = \hat{\mathbb{P}}_t^{\theta+\varepsilon \phi} \otimes d\hat{P}_t^\theta + \nu \varepsilon \phi_t \hat{\mathbb{P}}^{\theta+\varepsilon \phi}_t \otimes\word{2}~dt. 
\end{align*}
Therefore, we infer that 
\begin{align*}
    d\left[\frac{\hat{\mathbb{P}}^{\theta+\varepsilon \phi}_t-\hat{\mathbb{P}}^{\theta}_t}{\varepsilon}\right] = \frac{\hat{\mathbb{P}}^{\theta+\varepsilon \phi}_t-\hat{\mathbb{P}}^{\theta}_t}{\varepsilon} \otimes d\hat{{P}}^{\theta}_t +\nu  \phi_t \hat{\mathbb{P}}^{\theta+\varepsilon \phi}_t \otimes\word{2}~dt. 
\end{align*}
Let us now define $\mathbb{Y}_t$ and $\beta_t$ by
\begin{align*}
    &\mathbb{Y}_t:=\frac{\hat{\mathbb{P}}^{\theta+\varepsilon \phi}_t-\hat{\mathbb{P}}^{\theta}_t}{\varepsilon},\quad \beta_t = \nu  \phi_t \hat{\mathbb{P}}^{\theta+\varepsilon \phi}_t \otimes\word{2}.
\end{align*}
Then, we have that 
\begin{align*}
    d\mathbb{Y}_t = \mathbb{Y}_t~d\hat{{P}}^{\theta}_t+\beta_t dt. 
\end{align*}
In this case, we observe that 
\begin{align*}
    d\left(\mathbb{Y}_t\otimes (\hat{\mathbb{P}}^{\theta}_t)^{-1}\right) = d\mathbb{Y}_t\otimes(\hat{\mathbb{P}}^{\theta}_t)^{-1} + \mathbb{Y}_t\otimes d(\hat{\mathbb{P}}^{\theta}_t)^{-1}, 
\end{align*}
and thus
\begin{align*}
    \mathbb{Y}_t\otimes (\hat{\mathbb{P}}^{\theta}_t)^{-1} =  \int_0^t \beta_s \otimes\left(\hat{\mathbb{P}}^\theta_s\right)^{-1} ds.
\end{align*}
Thus, we have that 
\begin{align*}
    \mathbb{Y}_t &=\int_0^t \beta_s \otimes\left(\hat{\mathbb{P}}^\theta_s\right)^{-1} \otimes \left(\hat{\mathbb{P}}^{\theta}_t\right)  ds=\int_0^t \beta_s \otimes\hat{\mathbb{P}}^{\theta}_{s,t} ds=\nu\int_0^t \left( \phi_s \hat{\mathbb{P}}^{\theta+\varepsilon \phi}_s \otimes\word{2} \otimes\hat{\mathbb{P}}^{\theta}_{s,t} \right)ds.
\end{align*}
By the Fubini's theorem and dominated convergence theorem, since $\theta,\phi,\upsilon\in \mathcal{A}$ and $\alpha_t\in T^{p}(\mathbb{R}^2)$, we have that
\begin{equation*}
    \lim_{\varepsilon \to 0} \frac{1}{\varepsilon} \E\left(\upsilon_t\langle \alpha_t,\hat{\mathbb{P}}_t^{\theta+\varepsilon \phi}-\hat{\mathbb{P}}_t^{\theta}  \rangle \right)= \nu \int_0^t \lim_{\varepsilon \to 0} \E\left( \phi_s \upsilon_t\left\langle\alpha_t, \hat{\mathbb{P}}^{\theta+\varepsilon \phi}_s \otimes \word{2} \otimes \hat{\mathbb{P}}_{s,t}^\theta \right\rangle \right)ds. 
\end{equation*}
Finally, as $\theta,\phi,\upsilon\in \mathcal{A}$ and $\alpha_t\in T^{p}(\mathbb{R}^2)$, we observe that there exists $\tilde{p}>1$ such that, for all $s<t<T$,
\begin{equation*}
     \sup_{\varepsilon>0}\E\left( \left|\phi_s \upsilon_t\left\langle\alpha_t, \hat{\mathbb{P}}^{\theta+\varepsilon \phi}_s \otimes \word{2} \otimes \hat{\mathbb{P}}_{s,t}^\theta \right\rangle \right|^{\tilde{p}}\right)<\infty,
\end{equation*}
hence, $\left(\phi_s \upsilon_t\left\langle\alpha_t, \hat{\mathbb{P}}^{\theta+\varepsilon \phi}_s \otimes \word{2} \otimes \hat{\mathbb{P}}_{s,t}^\theta \right\rangle\right)_{\varepsilon>0}$ is $L^{\tilde{p}}$ bounded with $\tilde{p}>1$, and then uniformly integrable. Therefore, we deduce \eqref{eq: gateaux_diff_P}. 
\end{proof}

\subsection{Proof of Lemma \ref{lem: gateau_diff}}
\begin{proof}
First, we can rewrite the functional $J$ as
    \begin{align*}
        J(\theta)=&-\eta\langle \theta,\theta\rangle_{L_2} +  \langle\mu+\nu \theta,X^\theta- \langle \xi_.\proj{2}, \hat{\mathbb{P}}^\theta_. \rangle\rangle_{L_2} - \frac{\lambda}{2} \sigma^2 \langle X^\theta- \langle \xi_.\proj{2}, \hat{\mathbb{P}}^\theta_. \rangle,X^\theta- \langle \xi_.\proj{2}, \hat{\mathbb{P}}^\theta_. \rangle\rangle_{L_2} \\
        =&-\eta\langle \theta,\theta\rangle_{L_2}+\nu \langle \theta,X^\theta\rangle_{L_2} -\nu\langle\theta, \langle \xi.\proj{2}, \hat{\mathbb{P}}^\theta_. \rangle\rangle_{L_2} -\frac{\lambda}{2} \sigma^2 \langle X^\theta, X^\theta \rangle_{L_2}- \frac{\lambda}{2} \sigma^2 \langle\langle \xi_.\proj{2}, \hat{\mathbb{P}}^\theta_. \rangle,\langle \xi_.\proj{2}, \hat{\mathbb{P}}^\theta_. \rangle \rangle_{L_2}\\
        & + \lambda \sigma^2 \langle X^\theta,\langle \xi_.\proj{2}, \hat{\mathbb{P}}^\theta_. \rangle \rangle_{L_2}+\langle\mu,X^\theta- \langle \xi_.\proj{2}, \hat{\mathbb{P}}^\theta_. \rangle\rangle_{L_2} \\
        =:&\sum_{i=1}^7  J_i(\theta). 
    \end{align*}
For $i=1,...,7,$ let us compute the G\^ateaux derivatives of $J_i$ for $h\in \mathcal{A}$, such that 
\begin{equation*}
    \langle \nabla J_i(\theta), h\rangle_{L_2} := \lim_{\varepsilon\to 0} \frac{J_i(\theta+\varepsilon h)-J_i(\theta)}{\varepsilon}. 
\end{equation*}
First, we easily observe that 
\begin{align*}
    &\langle\nabla J_1(\theta),h\rangle_{L_2} =-2\eta \langle\theta,h \rangle_{L_2}, \\ 
    &\langle\nabla J_4(\theta),h\rangle_{L_2} =-\lambda \sigma^2 \langle\int_.^T\E(X^\theta_s|\mathcal{F}_t)ds,h \rangle_{L_2}. 
\end{align*}

Then, for $i=2$, we get that
\begin{align*}
    \frac{J_2(\theta+\varepsilon h)-J_2(\theta)}{\varepsilon}=&\frac{1}{\varepsilon} \Bigg[ \nu \E \left(\int_0^T \left((\theta_t + \varepsilon h_t) X_t^{\theta+\varepsilon h}\right) dt \right) - \nu \E \int_0^T \left(\theta_t X_t^\theta dt \right)\Bigg]\\
    =&\frac{1}{\varepsilon} \nu \left[\E \left(\int_0^T \varepsilon h_t X_t^{\theta+\varepsilon h} dt \right)+\E \left(\int_0^T \theta_t \left( X_t^{\theta+\varepsilon h}- X_t^\theta\right) dt \right)\right],
\end{align*}
thus, by the tower property of the conditional expectation and Fubini's theorem, taking $\lim_{\varepsilon \to 0}$ gives, 
\begin{align*}
    \langle\nabla J_2(\theta),h\rangle_{L_2}= \nu \langle X_.^\theta+\int_.^T \E(\theta_s|\mathcal{F}_.) ds,h\rangle_{L_2}.
\end{align*}
For $i=3$, 
\begin{align*}
    \frac{J_3(\theta+\varepsilon h)-J_3(\theta)}{\varepsilon}&=\frac{1}{\varepsilon}\nu \left[-\E\left(\int_0^T(\theta_t+\varepsilon h_t) \langle\xi_t\proj{2},\hat{\mathbb{P}}_t^{\theta+\varepsilon h}\rangle dt \right)+\E\left(\int_0^T\theta_t \langle\xi_t\proj{2},\hat{\mathbb{P}}_t^{\theta}\rangle dt \right)\right] \\
    &=-\frac{1}{\varepsilon} \nu\E \left(\int_0^T \theta_t \left(\langle \xi_t\proj{2},\hat{\mathbb{P}}_t^{\theta+\varepsilon h} \rangle-\langle \xi_t\proj{2},\hat{\mathbb{P}}_t^{\theta} \rangle\right)dt\right) -  \nu \langle \langle \xi_.\proj{2}, \hat{\mathbb{P}}^{\theta+\varepsilon h}_.\rangle, h\rangle_{L_2}.  
\end{align*}
Since $\theta, h\in \mathcal{A}$, using Lemma \ref{lem: gamma_der}, the dominated convergence theorem and Fubini's theorem, we obtain that 
\begin{align*}
     \lim_{\varepsilon \to 0} \frac{1}{\varepsilon} \E \left(\int_0^T \theta_t \left(\langle \xi_t\proj{2},\hat{\mathbb{P}}_t^{\theta+\varepsilon h} \rangle-\langle \xi_t\proj{2},\hat{\mathbb{P}}_t^{\theta} \rangle\right)dt\right) &=\E\left( \int_0^T h_t\int_t^T \E\left(\theta_s \nu \Gamma_{t,s}^\theta|\mathcal{F}_t\right)ds dt\right).
\end{align*}
\begin{comment}
Let us study the limit of \todo{Justify the passage to the limit under the integral sign by dominated convergence theorem?}
\begin{align*}
    \lim_{\varepsilon \to 0} \frac{1}{\varepsilon} \nu\E \left(\int_0^T \theta_t \left(\langle \xi_t\proj{2},\hat{\mathbb{P}}_t^{\theta+\varepsilon h} \rangle-\langle \xi_t\proj{2},\hat{\mathbb{P}}_t^{\theta} \rangle\right)dt\right) &= \lim_{\varepsilon \to 0} \frac{1}{\varepsilon}\nu \E \left(\int_0^T \theta_t \frac{\langle\xi_t\proj{2}, \hat{\mathbb{P}}_t^{\theta + \varepsilon h} \rangle-\langle\xi_t\proj{2}, \hat{\mathbb{P}}_t^\theta \rangle}{P_t^{\theta+\varepsilon h}-P_t^\theta} \left(\varepsilon \int_0^t h_s ds \right)dt\right) \\
    &=\lim_{\varepsilon \to 0} \nu \E \left(\int_0^T \theta_t \frac{\langle\xi_t\proj{2}, \hat{\mathbb{P}}_t^{\theta + \varepsilon h} \rangle-\langle\xi_t\proj{2}, \hat{\mathbb{P}}_t^\theta \rangle}{P_t^{\theta+\varepsilon h}-P_t^\theta} \nu \int_0^t h_s ds dt\right)\\
    &= \nu^2 \E \left(\int_0^T \theta_t \langle\xi_t\proj{22}, \hat{\mathbb{P}}_t^{\theta } \rangle\int_0^t h_s ds dt\right).\\
\end{align*}
\end{comment}
Therefore, we infer that 
\begin{align*}
    \langle\nabla J_3(\theta),h\rangle_{L_2} &=-\nu\left[\langle \int_.^T \E\left(\theta_s \nu \Gamma_{.,s}^\theta|\mathcal{F}_. \right) ds+  \langle \xi_.\proj{2}, \hat{\mathbb{P}}^{\theta}_.\rangle, h\rangle_{L_2} \right]. 
\end{align*}
Using similar arguments, we have that 
\begin{align*}
    \frac{J_5(\theta+\varepsilon h)-J_5(\theta)}{\varepsilon}
    &=-\frac{1}{\varepsilon}\frac{\lambda}{2} \sigma^2\E\left( \int_0^T \left(\langle (\xi_t\proj{2})\shupow{2}, \hat{\mathbb{P}}_t^{\theta+\varepsilon h}\rangle- \langle (\xi_t\proj{2})\shupow{2}, \hat{\mathbb{P}}_t^{\theta}\rangle \right) dt\right), 
\end{align*}
taking $\lim_{\varepsilon \to 0}$ leads to 
\begin{align*}
    \lim_{\varepsilon \to 0} \frac{J_5(\theta+\varepsilon h)-J_5(\theta)}{\varepsilon} &=- \frac{\lambda}{2} \sigma^2 \langle \int_.^T \E\left(\nu\tilde{\Gamma}_{.,s}^\theta |\mathcal{F}_.\right) ds,h\rangle_{L_2}.
\end{align*}
For $i=6$, we observe that 

\begin{align*}
    \frac{J_6(\theta+\varepsilon h)-J_6(\theta)}{\varepsilon}=&\lambda \sigma^2\bigg{[} \E \left(\int_0^T \langle\xi_t\proj{2}, \hat{\mathbb{P}}_t^{\theta+ \varepsilon h} \rangle \int_0^t h_s ds dt \right)\\
    &+\frac{1}{\varepsilon} \E \left(\int_0^T \left(\langle \xi_t\proj{2}, \hat{\mathbb{P}}_t^{\theta + \varepsilon h} \rangle - \langle \xi_t\proj{2}, \hat{\mathbb{P}}_t^{\theta} \rangle \right) X_t^\theta  dt \right) \bigg{]} 
\end{align*}
taking the limit $\lim_{\varepsilon \to 0}$ leads to 
\begin{align*}
    \lim_{\varepsilon \to 0} \frac{J_6(\theta+\varepsilon h)-J_6(\theta)}{\varepsilon}&=\lambda \sigma^2 \langle \int_.^T\E\left(\langle \xi_s\proj{2}, \hat{\mathbb{P}}_s^\theta \rangle + X_s^\theta \nu \Gamma_{.,s}^\theta|\mathcal{F}_.\right) ds,h\rangle_{L_2}. 
\end{align*}
Finally, for $i=7$, we have that
\begin{align*}
    \frac{J_7(\theta+\varepsilon h)-J_7(\theta)}{\varepsilon}&=\frac{1}{\varepsilon} \mu \left[\E \int_0^T \left((X_t^{\theta +\varepsilon h}-X_t^\theta)-(\langle \xi_t\proj{2}, \hat{\mathbb{P}}_t^{\theta+\varepsilon h}\rangle -\langle \xi_t\proj{2}, \hat{\mathbb{P}}_t^{\theta}\rangle ) \right) dt \right],
\end{align*}
and using similar arguments, we obtain that
\begin{align*}
    \langle\nabla J_7(\theta),h\rangle_{L_2} = \mu \langle (T-.)-\int_.^T \E \left(\nu \Gamma_{.,s}^\theta|\mathcal{F}_.\right) ds, h\rangle_{L_2}.  
\end{align*}
Taking all together, we finally get that 
\begin{align*}
    \langle \nabla J(\theta), h\rangle_{L_2} =\sum_{i=1}^7 \langle \nabla J_i(\theta), h\rangle_{L_2} = \langle -2\eta \theta+\nu \mathbf{A}\theta-\lambda \sigma^2 \mathbf{B}\theta + \mu \mathbf{C}\theta,h\rangle_{L_2}, 
\end{align*}
with the operators $\mathbf{A}$, $\mathbf{B}$ and $\mathbf{C}$ defined by \eqref{eq: op_A} and \eqref{eq: op_B}.
\end{proof}

\subsection{Proof of Proposition \ref{prop: example_monotonicity_ok_without_permanent}}
Since $\nu=0$, then $p=2$ and then the monotonicity condition reduces to 
    \begin{equation*}
        \left \langle -(2\eta-\varepsilon)(\theta-\phi)-\lambda \sigma^2 \left(\mathbf{B}(\theta)-\mathbf{B}(\phi) \right) + \mu \left(\mathbf{C}(\theta)-\mathbf{C}(\phi) \right), \theta-\phi \right \rangle_{L_2} \leq 0. 
    \end{equation*}
    Moreover, we have  that
    \begin{align*}
        \left \langle (\mathbf{B}(\theta)-\mathbf{B}(\phi)), \theta-\phi \right \rangle_{L_2} &= \E\left(\int_0^T (X^{\theta}_t-X^{\phi}_t)^2dt \right)
\quad \text{and} \quad 
        \left \langle \mu (\mathbf{C}(\theta)-\mathbf{C}(\phi)), \theta-\phi \right \rangle_{L_2} = 0. 
    \end{align*}
    Hence, for any $\varepsilon>0$ such that $0<\varepsilon< 2\eta$, we infer that 
    \begin{align*}
        &\left \langle -(2\eta-\varepsilon)(\theta-\phi)-\lambda \sigma^2 \left(\mathbf{B}(\theta)-\mathbf{B}(\phi) \right) + \mu \left(\mathbf{C}(\theta)-\mathbf{C}(\phi) \right), \theta-\phi \right \rangle_{L_2}\\
        &\leq \left(-(2\eta-\varepsilon)\right)\E\left(\int_0^T (\theta_t-\phi_t)^2dt \right)-\lambda\sigma^2 \E\left(\int_0^T (X_t^\theta-X_t^\phi)^2dt \right)\\
        &\leq  0
    \end{align*}
    and the monotonicity condition \eqref{eq:monotonicity_cond} is satisfied. 

\subsection{Proof of Proposition \ref{prop: example_monotonicity_ok}}
\begin{proof}
    First, since $\tilde{M}\leq 1$, then $p=2$ and the monotonicity condition reduces to 
    \begin{equation*}
        \left \langle -(2\eta-\varepsilon)(\theta-\phi)+\nu \left(\mathbf{A}(\theta)-\mathbf{A}(\phi)\right)-\lambda \sigma^2 \left(\mathbf{B}(\theta)-\mathbf{B}(\phi) \right) + \mu \left(\mathbf{C}(\theta)-\mathbf{C}(\phi) \right), \theta-\phi \right \rangle_{L_2} \leq 0. 
    \end{equation*}
    Moreover, since $\xi_t\proj{2}\in T^{\leq 1}(\mathbb{R}^2)$, we observe that
    \begin{align*}
        &\Gamma_{t,s}^\theta= \xi_s^{\word{22}}, \quad \tilde{\Gamma}_{t,s}^\theta= 2 \xi_s^{\word{2}} \xi_s^{\word{22}}+ 2 P_s^\theta \left(\xi_s^{\word{22}}\right)^2 + 2 \xi_s^{\word{12}} \xi_s^{\word{21}} s. 
    \end{align*}
    In this case, by Cauchy-Schwarz inequality, as we assume \eqref{eq:condition_time_dependent_gamma}, we infer that
    \begin{align*}
        \left \langle \nu (\mathbf{A}(\theta)-\mathbf{A}(\phi)), \theta-\phi \right \rangle_{L_2} \leq   \frac{(2\eta-\varepsilon)}{\nu}\E\left(\int_0^T (\theta_t-\phi_t)^2dt \right),
    \end{align*}
    as well as 
    \begin{align*}
        \left \langle  (\mathbf{B}(\theta)-\mathbf{B}(\phi)), \theta-\phi \right \rangle_{L_2} &= \E\left(\int_0^T (1-\nu \xi_t^{\word{22}})^2(X^{\theta}_t-X^{\phi}_t)^2dt \right),
    \end{align*}
    and finally
    \begin{align*}
        \left \langle \mu (\mathbf{C}(\theta)-\mathbf{C}(\phi)), \theta-\phi \right \rangle_{L_2} &= 0. 
    \end{align*}
    Hence, we readily infer that 
    \begin{align*}
        &\left \langle -(2\eta-\varepsilon)(\theta-\phi)+\nu \left(\mathbf{A}(\theta)-\mathbf{A}(\phi)\right)-\lambda \sigma^2 \left(\mathbf{B}(\theta)-\mathbf{B}(\phi) \right) + \mu \left(\mathbf{C}(\theta)-\mathbf{C}(\phi) \right), \theta-\phi \right \rangle_{L_2}\\
        &\leq -\lambda\sigma^2 \E\left(\int_0^T (1-\nu \xi_t^{\word{22}})^2(X_t^\theta-X_t^\phi)^2dt \right)\\
        &\leq  0,
    \end{align*}
    and therefore the monotonicity condition \eqref{eq:monotonicity_cond} is satisfied. 
\end{proof}

\section*{Conflicts of Interest}
The authors declare no conflicts of interest.

\section*{Data Availability Statement}
Data sharing is not applicable to this article as no new data were created or analyzed in this study. 

\bibliographystyle{plainnat}
\bibliography{bibl}

\end{document}